\documentclass[letterpaper,11pt]{article}

\usepackage[T1]{fontenc}
\usepackage[utf8]{inputenc}
\usepackage{authblk}
\usepackage{amsthm}
\usepackage{amsmath}
\usepackage{amsfonts}
\usepackage{verbatim}
\usepackage{tikz}
\usepackage{multirow}

\newtheorem{lemma}{Lemma}
\newtheorem{theorem}{Theorem}
\newtheorem{definition}{Definition}

\newtheorem{conjecture}{Conjecture}

\newtheorem{claim}{Claim}
\newenvironment{claimproof}{{\bf Proof of the Claim:}}{\hspace*{\fill} \bf\tt End of Claim Proof}

\title{ \vspace*{-0.5in} Relaxed Byzantine Vector Consensus\thanks{This research is supported in part by National Science Foundation award 1421918. Any opinions, findings, and conclusions or recommendations expressed here are those of the authors and do not necessarily reflect the views of the funding agency or the U.S. government.}}
\author[1]{Zhuolun Xiang\thanks{xiangzhuolun@gmail.com}}
\author[2]{Nitin H.Vaidya\thanks{nhv@illinois.edu}}
\affil[1]{Institute for Interdisciplinary Information Sciences\protect\linebreak Tsinghua University}
\affil[2]{Department of Electrical and Computer Engineering\protect\linebreak University of Illinois at Urbana-Champaign}

\begin{document}
  \maketitle

\begin{abstract}
Exact Byzantine consensus problem requires that non-faulty processes reach agreement on a decision (or output) that is in the convex hull of the inputs at the non-faulty processes. It is well-known that exact consensus is
impossible in an asynchronous system in presence of faults, and in a synchronous system, the number of processes $n$ must be at least $3f+1$ to be able to achieve exact Byzantine consensus with scalar inputs, in presence of up to $f$ Byzantine faulty processes.
Recent work has shown that when the inputs are $d$-dimensional vectors of reals, $n\geq \max(3f+1,(d+1)f+1)$ is the tight bound on the number of processes $n$
to be able to achieve exact Byzantine consensus in a synchronous system. 
In an asynchronous system approximate Byzantine consensus is possible if and only if $n\geq (d+2)f+1$.

Due to the dependence of the lower bound on vector dimension $d$, the number
of processes necessary becomes large when the vector dimension is large.
With the hope of reducing the lower bound on $n$,
we consider two relaxed versions of Byzantine vector consensus: \textsl{k-Relaxed Byzantine vector consensus} and \textsl{$(\delta,p)$-Relaxed Byzantine vector consensus}. 
In $k$-relaxed Byzantine consensus, the validity condition requires that the output must be in the convex hull
of projection of the inputs onto any subset of $k$-dimensions of the vectors.
For $(\delta,p)$-consensus the validity condition requires that the output must be
within distance $\delta$ of the convex hull of the inputs
of the non-faulty processes, where $L_p$ norm is used as the distance metric. 
For $(\delta,p)$-consensus, we consider two versions: in one version, $\delta$ is a constant,
and in the second version, $\delta$ is a function of the inputs themselves.

We show that for
$k$-relaxed consensus and $(\delta,p)$-consensus with constant $\delta\geq 0$,
the bound on $n$ is identical to the bound stated above for the original
vector consensus problem.
On the other hand, when $\delta$ depends on the inputs, we show that the bound on $n$ is smaller
when $d\geq 3$.

\end{abstract}

\section{Introduction}

This paper considers Byzantine consensus in a complete network consisting of $n$ processes
of which up to $f$ processes may be Byzantine faulty \cite{lamport1982byzantine}.
$n$ is assumed to be at least 2, since consensus is trivial for $n=1$.
The exact Byzantine consensus problem requires that non-faulty processes reach agreement on an identical decision (or output) that is in the convex hull of the inputs at the non-faulty processes. It is well-known that exact consensus is
impossible in an asynchronous system in presence of faults \cite{fischer1985impossibility}. In a synchronous system, the number of processes $n$ must be at least $3f+1$ to be able to achieve exact Byzantine consensus with scalar inputs \cite{fischer1990easy}.
Recent work
 has shown that when the inputs are $d$-dimensional vectors of reals, $n\geq \max(3f+1,(d+1)f+1)$ is the tight bound on the number of processes $n$
to be able to achieve exact Byzantine consensus in a synchronous system
 \cite{vaidya2013byzantine}.
In an asynchronous system, it is shown that approximate Byzantine consensus is possible if and only if $n\geq (d+2)f+1$
  \cite{mendes2013multidimensional,vaidya2013byzantine}.

Due to the dependence of the lower bound on vector dimension $d$, the number
of processes necessary becomes large when the vector dimension is large.
With the hope of reducing the lower bound on $n$,
we consider two relaxed versions of Byzantine vector consensus: \textsl{k-Relaxed Byzantine vector consensus} and \textsl{$(\delta,p)$-Relaxed Byzantine vector consensus}. 
In $k$-relaxed Byzantine consensus, the validity condition requires that the output must be in the convex hull
of projection of the inputs onto any subset of $k$-dimensions of the vectors.
For $(\delta,p)$-consensus the validity condition requires that the output must be
within distance $\delta$ of the convex hull of the inputs
of the non-faulty processes, where $L_p$ norm is used as the distance metric. 
For brevity, we will often refer to these as $k$-consensus or $k$-relaxed consensus,
and ($\delta,p$)-consensus or ($\delta,p$)-relaxed consensus, respectively.
For $(\delta,p)$-consensus, we consider two versions: in one version, $\delta$ is a constant,
and in the second version, $\delta$ is a function of the inputs themselves.
Note that the vector consensus problem defined in \cite{mendes2013multidimensional,vaidya2013byzantine} is obtained as a special case of the above relaxed versions. In particular,
the previous problem is identical to $n$-relaxed consensus and ($0,2$)-relaxed consensus.



The main contributions of this paper are as follows:
\begin{itemize}
\item For synchronous and asynchronous systems both, we show that the tight bound on $n$ for $k$-relaxed consensus is identical for all $k$ such that $1<k\leq d$. That is, when $k>1$, the relaxation does not reduce the number of processes necessary. When $k=1$, $n\geq 3f+1$ is necessary and sufficient for all dimensions $d$.


\item For synchronous and asynchronous systems both, for a {\em constant} $\delta\geq 0$, we show that the tight bound on $n$ for $(\delta,p)$-relaxed consensus is identical to that for $\delta=0$
and $p=2$. That is, the relaxation does not reduce the number of processes necessary when $\delta$
is a constant.

\item For certain values of $\delta$ specified as a function of the inputs
of the non-faulty processes, we show that
 $(\delta,p)$-consensus can be achieved in both synchronous and asynchronous systems
with a smaller number of processes. 
We establish a relationship between $n$ and an achievable
value of $\delta$. For instance, for $f=1$ and $d\geq 3$, we show that $(\frac{e_{max}}{d-1},2)$-consensus and $(\frac{e_{min}}{2}$, 2)-consensus is achievable with $n=(d+1)$ processes, where $e_{max}$ ($e_{min}$) is the maximum (minimum) distance between the inputs of any two fault-free processes.
We also obtain partial results for other values of $f$, $n$ and $p$, and propose a conjecture for one remaining case.



\end{itemize}

\section{Related Work}

The necessary and sufficient condition for consensus in presense of Byzantine failure under various underlying network conditions is extensively studied in the literature. Lamport, Shostak and Pease \cite{lamport1982byzantine} developed the initial results on Byzantine fault-tolerant agreement. The FLP result \cite{fischer1985impossibility} showed that exact consensus is impossible even under single process failure in asynchronous systems. To circumvent this obstacle,
Dolev et al. \cite{dolev1986reaching} proposed \textsl{approximate consensus} for asynchronous systems. 

When $d=1$, the inputs are scalar, and all the $L_p$ norms are identical. For the case of $d=1$,
$(\delta,p)$-relaxed consensus is equivalent to a problem that was addressed in
prior work \cite{fischer1990easy}; for this special case, it was shown that
$n\geq 3f+1$ is necessary and sufficient \cite{fischer1990easy}.

The
\textsl{Byzantine vector consensus} (BVC) problem (also called {\em multidimensional} consensus) was introduced by Mendes and Herlihy \cite{mendes2013multidimensional}
and Vaidya and Garg \cite{vaidya2013byzantine}. Tight bounds
on number of processes $n$ for Byzantine vector consensus have been obtained
for synchronous \cite{vaidya2013byzantine} 
 and asynchronous \cite{mendes2013multidimensional,vaidya2013byzantine} systems both,
when the network is a complete graph. A necessary condition and a sufficient condition for {\em iterative} byzantine vector consensus were derived by Vaidya \cite{vaidya2014iterative}, however, there is a gap between these necessary and sufficient conditions.

A more generalized problem called Convex Hull Consensus problem was introduced by Tseng and Vaidya \cite{tseng2014asynchronous}. The tight bounds on number of processes $n$ is identical to the vector consensus case. Optimal fault resilient algorithms were proposed for asynchronous systems under crash faults \cite{tseng2014asynchronous} and Byzantine faults \cite{tseng2013byzantine}, respectively.

Herlihy et al. \cite{herlihy2014computing} study a new version of the approximate vector consensus problem, called \textsl{$(d, \epsilon)$-solo approximate agreement}, in the context of a $d$-solo
execution model that yields the message-passing model and the traditional shared memory model as special cases.
For $(d,\epsilon)$-solo approximate agreement, the inputs are $d$-dimensional vectors of reals,
and the outputs must be in the convex hull of all the inputs.
Up to $d$ processes may potentially choose as their ouputs any arbitrary
 points in the convex hull of all inputs
 (not necessarily approximately equal to each other),
while each remaining process must choose as its output a point within distance $\epsilon$ of the convex hull of the outputs of these $d$ processes (all outputs must be within the convex hull of the inputs). Although Herlihy et al. \cite{herlihy2014computing} only consider crash failures, the problem can be easily extended to the Byzantine fault model.
The relaxed consensus formulations considered in our work are different from
$(d,\epsilon)$-solo agreement.

\section{Notations and Terminology}

The network is assumed to be a complete graph, i.e., there is a reliable communication
channel from every process to each of the remaining processes.
The total number of processes is $n$, with up to $f$ processes suffering
Byzantine failures.
The input at each process is a $d$-dimensional vector, $d\geq 1$.
We will index the dimensions of a $d$-dimensional vector as as $1,2,\cdots,d$.
We will view the inputs as {\em column} vectors. Transpose of vector $u$ will be denoted as $u^T$.
We will also often view a {\em vector} as a {\em point} in an appropriate space.
The $i$-th element (or $i$-th coordinate) of vector $v$ will be denoted as $v[i]$.
We denote the set $\{1,2,\cdots,d\}$ by $[1,d]$.
For $u,v\in \mathbb{R}^d$, distance $\|u-v\|_p$ using $L_p$-norm is defined as 
$$\|u-v\|_p = \left(\sum_{i=1}^d |u[i]-v[i]|^p\right)^{1/p}$$.

A multiset may potentially contain repetitions of an element.
For instance, $\{1,1,3,5\}$ is a multiset in which value 1 is repeated.
Similarly, given $k$-dimensional vectors, $u,v,w$, $\{u,v,v,w,w,w\}$ is a multiset.
The standard set, in which each unique element appears at most once, is a special
case of a multiset.
Let $$\mathcal{H}(S)$$ denote the convex hull of a multiset $S$.
For a multiset $Y$, when we write $T\subseteq Y$, $T$ is a multiset
in which frequency of each element is no greater than its frequency in multiset $Y$.
Thus, $\{u,v,v,w,w\}\subseteq \{u,v,v,w,w,w\}$.
Size of the multiset $S$, denoted $|S|$, is the number of elements in $S$,
counting all repetitions. Thus, $|\{u,v,v,w,w,w\}|=6$. For a multiset $Y$ with
$|Y|\geq f$, define $\Gamma(Y)$
as follows.
$$\Gamma(Y)=\bigcap_{T\subseteq Y, |T|=|Y|-f}\mathcal{H}(T)$$.

\section{Previous Result}
\label{sec:previous}

The problem of \textsl{Byzantine vector consensus} (BVC) in complete graphs was
studied in \cite{mendes2013multidimensional,vaidya2013byzantine}. Here we briefly summarize the previous results.
The input of each process is assumed to be a $d$-dimensional vector of reals.

\textbf{Exact BVC:} Exact Byzantine vector consensus must satisfy the following three conditions \cite{vaidya2013byzantine}:
\begin{enumerate}
	\item\textsl{Agreement:} The decision (or output) vector at all the non-faulty processes must be identical.
	\item\textsl{Validity:} The decision vector at each non-faulty process must be in the convex hull of the input vectors at the non-faulty processes.
	\item\textsl{Termination:} Each non-faulty process must terminate after a finite amount of time.
\end{enumerate}

\textbf{Approximate BVC:} Approximate Byzantine vector consensus must satisfy the following three conditions \cite{mendes2013multidimensional,vaidya2013byzantine}:
\begin{enumerate}
	\item\textsl{$\epsilon-$Agreement:} The decision vectors at any two non-faulty processes must be within distance $\epsilon$ of each other, where $\epsilon>0$.

	\item\textsl{Validity:} The decision vector at each non-faulty process must be in the convex hull of the input vectors at the non-faulty processes.
	\item\textsl{Termination:} Each non-faulty process must terminate after a finite amount of time.
\end{enumerate}
For the distance in the $\epsilon$-agreement condition above, \cite{mendes2013multidimensional} uses the $L_2$ norm (or Euclidean distance) and \cite{vaidya2013byzantine} uses the $L_\infty$ norm (bounding the maximum difference in each vector coordinate
of the outputs). While the bound on $n$ is not affected by the choice of the norm in the prior work, the choice of the norm does affect some of our results (particularly, the results for $(\delta,p)$-consensus when $\delta$ depends on the inputs).

The following results have been obtained previously.

\begin{theorem}
\label{thm:exact}
	$n \geq \max(3f + 1, (d + 1)f + 1)$ is necessary and sufficient for Exact BVC in a synchronous system
 \cite{vaidya2013byzantine}.
\end{theorem}

\begin{theorem}
\label{thm:approximate}
	$n\geq (d + 2)f + 1$ is necessary and sufficient for Approximate BVC in an asynchronous system
 \cite{mendes2013multidimensional,vaidya2013byzantine}.
\end{theorem}

\section{Relaxed Byzantine Vector Consensus}

This section
formally defines $k$-relaxed consensus and $(\delta,p)$-relaxed consensus,
where $\delta$ is either a constant or depends on the inputs.


\subsection{k-Relaxed Consensus}

To be able to define $k$-relaxed consensus, we first need to present some
other definitions.
Given a size $k$ subset $D$ of $[1,d]$, we define a projection function $g_D$ below.
For any $d$-dimensional vector $u$, $g_D$ yields a $k$-dimensional vector $v$
retaining only those elements of $u$ whose indices are included in $D$.

\begin{definition}
	Let $D = \{d_1, d_2,\cdots, d_k\}$ where $1\leq d_i < d_j\leq d$ for $1\leq i<j\leq k$.
	For $u \in \mathbb{R}^d$ define projection $g_D(u) = v$ where $v\in \mathbb{R}^k$
	and $v[i]=u[d_i]$.
\end{definition}

We will also refer to $g_D$ as the $D$-projection.
For example, suppose that $d=4$, $D=\{1,3\}$ and $u=(7,-4,-2,0)^T$.  Then $g_D(u)=(7,-2)^T$.

\begin{definition}
Define $\mathcal{D}_k$ to be the set of all size $k$ subsets of $[1,d]$. That is, $$\mathcal{D}_k=\{ D~|~D \subseteq [1,d],~|D|=k\}.$$
\end{definition}

While $g_D$ is not a one-to-one function, with an abuse of terminology, we will define
its inverse.
Inverse of $g_D$, namely $g_D^{-1}$, maps each $k$-dimensional
vector $v$ to the set of $d$-dimensional vectors whose $D$-projection is $v$.

\begin{definition}
	For $D\in\mathcal{D}_k$ and
	$v \in \mathbb{R}^k$, define $g_D^{-1}(v) = U$ where $U\subset \mathbb{R}^d$,
	such that $u\in U$ if and only if $g_D(u)=v$.
\end{definition}
For example, suppose that $d=4$, $D=\{1,3\}$ and $v^T=(7,-2)$.
Then $g_D^{-1}(v)=\{ (7,a,-2,b)~|~ a,b \in \mathbb{R}\}$.
We will use the shorthand $g_D^{-1}(v) = (7,*,-2,*)^T$ to denote the above set,
which consists of all 4-dimensional column vectors whose first element is 7 and
third element is $-2$.

We now define function $g_D$ with a multiset of points $S$ in the $d$-dimensional space
as its argument. $g_D(S)$ is also a multiset. With some liberty with terminology, we will define
$g_D(S)$ using $g_D(u)$ defined previously with a single point as the argument.
\begin{definition}
	For $D\in\mathcal{D}_k$ and
	multiset $S$ consisting of points in  $\mathbb{R}^d$, define $g_D(S) = \{g_D(u)~|~u\in S\}$.
	$g_D(S)$ is a multiset.
\end{definition}
The inverse function $g_D^{-1}$ is similarly extended to multisets of points
in the $k$-dimensional space.

\begin{definition}
For $D\in\mathcal{D}_k$ and
	multiset $S$ consisting of points in $\mathbb{R}^k$, define $$g_D^{-1}(S) = \bigcup_{v\in S}~ g_D^{-1}(v).$$
\end{definition}

\begin{definition}
	k-relaxed convex hull $H_k$ of $S \subset \mathbb{R}^d$ is defined as
	$$
	H_k(S)= \{u~|~g_D(u)\in \mathcal{H}(g_D(S)), \forall D \in \mathcal{D}_k  \}
	$$
	Equivalently,
	$$H_{k}(S)=\bigcap_{D\in \mathcal{D}_k} ~g_D^{-1}(\mathcal{H}(g_D(S)))$$
\end{definition}

Now we can formally define $k$-relaxed consensus.

\begin{definition}[$k$-Relaxed Exact BVC]
	k-Relaxed exact Byzantine vector consensus must satisfy the following three conditions.
	\begin{enumerate}
		\item\textsl{Agreement:} The decision (or output) vector at all the non-faulty processes must be identical.
		\item\textsl{k-Relaxed Validity:} The decision vector at each non-faulty process must be in the \textsl{k-relaxed convex hull} of the set of input vectors at the non-faulty processes.
		\item\textsl{Termination:} Each non-faulty process must terminate after a finite amount of time.
	\end{enumerate}
\end{definition}

\begin{definition}[$k$-Relaxed Approximate BVC]
	k-Relaxed approximate Byzantine vector consensus must satisfy the following three conditions.
	\begin{enumerate}
		\item\textsl{$\epsilon-$Agreement:} For $1\leq l\leq d$, the $l^{th}$ elements of the decision vectors at any two non-faulty processes must be within $\epsilon$ of each other, where $\epsilon>0$ is a pre-defined constant.\footnote{Effectively,
this definition uses $L_\infty$-norm to define distance between output vectors,
and bounds this distance by $\epsilon$. However, the bounds on $n$ hold
for any $L_p$-norm, $p\geq 1$, due to norm equivalence (as discussed later).
}

		\item\textsl{k-Relaxed Validity:} The decision vector at each non-faulty process must be in the \textsl{k-relaxed convex hull} of the set of input vectors at the non-faulty processes.
		\item\textsl{Termination:} Each non-faulty process must terminate after a finite amount of time.
	\end{enumerate}
\end{definition}

\subsection{$\boldsymbol{(\delta,p)}$-Relaxed BVC}

We define $(\delta,p)$-relaxed consensus using another relaxed notion of a convex hull.

\begin{definition}
\label{def:Hdelta}
	For $\delta\geq 0$ and $p\geq 1$, $(\delta,p)$-relaxed convex hull $H_{(\delta,p)}$ of $S\subseteq \mathbb{R}^d$ is
	$$
	H_{(\delta,p)}(S)= \{u~|~\|u-v\|_p\leq \delta,~v\in \mathcal{H}(S)\}
	$$
\end{definition}

Now we define exact and approximate $(\delta,p)$-relaxed consensus. These definitions are
independent of
whether $\delta$ is a constant, or depends on the inputs at non-faulty processes.

\begin{definition}[${(\delta,p)}$-Relaxed Exact BVC, $\delta\geq 0$]
	$(\delta,p)$-Relaxed exact Byzantine vector consensus must satisfy the following three conditions.
	\begin{enumerate}
		\item\textsl{Agreement:} The decision (or output) vector at all the non-faulty processes must be identical.
		\item\textsl{$(\delta,p)$-Relaxed Validity:} The decision vector at each non-faulty process must be in the \textsl{$(\delta,p)$-relaxed convex hull} of the input vectors at the non-faulty processes.
		\item\textsl{Termination:} Each non-faulty process must terminate after a finite amount of time.
	\end{enumerate}
\end{definition}

\begin{definition}[${(\delta,p)}$-Relaxed Approximate BVC, $\delta\geq 0$]
	$(\delta,p)$-Relaxed approximate Byzantine vector consensus must satisfy the following three conditions.
	\begin{enumerate}
		\item\textsl{$\epsilon-$Agreement:} For $1\leq l\leq d$, the $l^{th}$ elements of the decision vectors at any two non-faulty processes must be within $\epsilon$ of each other, where $\epsilon>0$ is a pre-defined constant.\footnote{
This definition uses $L_\infty$-norm to define distance between output vectors,
and bounds this distance by $\epsilon$. However, the bounds on $n$ hold
for any $L_q$-norm, $ q\geq 1$, due to norm equivalence (as discussed later).
On the other hand, the choice of $p$ in defining
the $(\delta,p)$-relaxed convex hull does affect some of the results, as also discussed
later.}



		\item\textsl{$(\delta,p)$-Relaxed Validity:} The decision vector at each non-faulty process must be in the \textsl{$(\delta,p)$-relaxed convex hull} of the input vectors at the non-faulty processes.
		\item\textsl{Termination:} Each non-faulty process must terminate after a finite amount of time.
	\end{enumerate}
\end{definition}

\subsection{Discussion}
\textsl{k-Relaxed BVC} and \textsl{$(\delta,p)$-Relaxed BVC} are both relaxed version of \textsl{BVC}. As we can see, both relaxed convex hulls of a set of points $S$ contain the convex hull  $\mathcal{H}(S)$ of $S$; therefore, solutions of the original BVC problem discussed
in Section \ref{sec:previous} are also solutions to \textsl{k-Relaxed BVC} and \textsl{$(\delta,p)$-Relaxed BVC}.

For $k$-relaxed consensus, notice that when $k=d$, the problem becomes the
same as the original BVC problem. So the necessary and sufficient conditions to the problem are known when $k=d$: $n\geq \max(3f+1,(d+1)f+1)$ is necessary and sufficient for \textsl{d-Relaxed Exact BVC}, and $n\geq (d+2)f+1$ is necessary and sufficient for \textsl{d-Relaxed Approximate BVC}.
When $k=1$, the $k$-relaxed consensus (i.e., 1-relaxed BVC) 
can be achieved using Byzantine scalar consensus as follows.
Each process chooses the $i$-th coordinate of the output vector
as the output of the scalar Byzantine consensus algorithm for
which the input of each process is the $i$-th coordinate of its input vector.
It is easy to verify that this solves 1-relaxed consensus.
Therefore,
the bounds for $k=1$ follow from previous results: $n\geq 3f+1$
is necessary and sufficient for \textsl{1-Relaxed Exact BVC} and \textsl{1-Relaxed Approximate BVC} both.

When $\delta=0$,
the \textsl{$(\delta,p)$-relaxed BVC} becomes identical to the original BVC problem
in Section \ref{sec:previous}. So the necessary and sufficient conditions for solving the problem are known: $n\geq \max(3f+1,(d+1)f+1)$ for \textsl{$(0,p)$-Relaxed Exact BVC} and $n\geq (d+2)f+1$ for \textsl{$(0,p)$-Relaxed Approximate BVC}. When $\delta=\infty$, the
validity condition is vacuous, allowing the processes to choose any fixed vector
in $\mathbb{R}^d$ as the output (e.g., the processes may always choose the all-0 vector
as their output and still satisfy the validity condition with $\delta=\infty$).

As we will see soon, our results are somewhat
disappointing: the tight bound on $n$ for both \textsl{k-Relaxed BVC} where $2\leq k \leq d-1$ and \textsl{$(\delta, p)$-Relaxed BVC}, where
$\delta$ is constant, $0<\delta<\infty$ and $p\geq 1$, is not lower than the original
formulations of the exact and approximate BVC problems, respectively.
However, when $\delta$ may depend on the inputs, we obtain a weaker requirement on
the number of processes.

\subsection{Useful Lemmas}

The following lemmas can be proved easily. Some of the proofs are omitted here for brevity.

\begin{lemma}
\label{lemma:Hk}
	For $S\subset \mathbb{R}^d$, the following
containment order holds for the k-relaxed convex hulls:
	$$
	H_{i}(S)\subseteq H_{j}(S), ~~ d \geq i\geq j\geq 1
	$$
\end{lemma}
\begin{proof}
	By the definition of k-relaxed convex hull, we have
 $$H_i(S)=\bigcap_{I\in \mathcal{D}_i}g_I^{-1}(\mathcal{H}(g_I(S)))$$
and
 $$H_j(S)=\bigcap_{J\in \mathcal{D}_j}g_J^{-1}(\mathcal{H}(g_J(S))).$$
For all $I\in \mathcal{D}_i$ and $J\in \mathcal{D}_j$ such that
$J\subseteq I$,
it is true that $g_I^{-1}(\mathcal{H}(g_I(S)))\subseteq g_J^{-1}(\mathcal{H}(g_J(S)))$.
Then the lemma follows from the above expressions for $H_i(S)$ and $H_j(S)$.
\end{proof}

\begin{lemma}
\label{lemma:k:1}
A necessary condition for k-Relaxed Exact BVC, $1\leq k<d$, is also necessary for (k+1)-Relaxed Exact BVC.
\end{lemma}
\begin{proof}
Suppose that $S$ is the set of inputs at non-faulty processes.
By Lemma \ref{lemma:Hk}, for $S\subset \mathbb{R}^d$, $H_{k+1}(S)\subseteq H_k(S)$.
Thus, for a given set of inputs at non-faulty processes,
if $k$-consesus is not achieved, then $(k+1)$-consensus is also not achieved.
The lemma then follows.
\end{proof}
The next three lemmas follow using similar arguments as the above proof. We omit their proofs for brevity.

\begin{lemma}
\label{lemma:k:2}
A sufficient condition for (k+1)-Relaxed Exact BVC, $1\leq k<d$, is also sufficient for k-Relaxed Exact BVC.
\end{lemma}

\begin{lemma}
\label{lemma:k:3}
	A necessary condition for k-Relaxed Approximate BVC, $1\leq k<d$, is also necessary for (k+1)-Relaxed Approximate BVC.
\end{lemma}

\begin{lemma}
\label{lemma:k:4}
	A sufficient condition for (k+1)-Relaxed Approximate BVC, $1\leq k<d$, is also sufficient for k-Relaxed Approximate BVC.
\end{lemma}

Now we show similar relationships for $(\delta,p)$-relaxed consensus.

\begin{lemma}
\label{lemma:delta:1}
	A necessary condition for $(\delta,p)$-Relaxed Exact BVC is also necessary for $(\delta',p)$-Relaxed Exact BVC, where $\delta\geq \delta'\geq 0$.
\end{lemma}
\begin{proof}
Suppose that $S$ is the set of inputs at non-faulty processes.
By Definition \ref{def:Hdelta}, $H_{(\delta',p)}(S)\subseteq H_{(\delta,p)}(S)$. Thus, for a given set of inputs
at non-faulty processes, if exact $(\delta',p)$-consensus is not achieved,
then exact $(\delta,p)$-consensus is also not achieved. The lemma then follows.
\end{proof}

The next three lemmas below can be proved similarly. Their proofs
are omitted for brevity.

\begin{lemma}
\label{lemma:delta:2}
	A sufficient condition for $(\delta',p)$-Relaxed Exact BVC is also sufficient for $(\delta,p)$-Relaxed Exact BVC, where $\delta\geq \delta'\geq 0$.
\end{lemma}

\begin{lemma}
\label{lemma:delta:3}
	A necessary condition for $(\delta,p)$-Relaxed Approximate BVC is also necessary for $(\delta',p)$-Relaxed Approximate BVC, where $\delta\geq \delta'\geq 0$.
\end{lemma}

\begin{lemma}
\label{lemma:delta:4}
	A sufficient condition for $(\delta',p)$-Relaxed Approximate BVC is also sufficient for $(\delta,p)$-Relaxed Approximate BVC, where $\delta\geq \delta'\geq 0$.
\end{lemma}

\section{k-Relaxed Byzantine Vector Consensus}
\label{sec:k}

\subsection{Synchoronous Systems}
In this section, we prove the necessary and sufficient condition for \textsl{k-Relaxed Exact BVC} in a synchronous system, where $2\leq k\leq d-1$ (thus, $d\geq 3$). 
As noted earlier, for $k=1$, $n\geq 3f+1$ is necessary and sufficient,
and for $k=d$, $n\geq \max(3f+1,(d+1)f+1)$ is necessary and sufficient.
Bounds for $d=1,2$ are included in the above results.

\begin{theorem}
\label{thm:k:sync}
	$n\geq (d+1)f+1$ is necessary and sufficient for k-Relaxed Exact BVC in a synchronous system when $2\leq k\leq d-1$.
\end{theorem}

\begin{proof}
Since $2\leq k\leq d-1$, we have $d\geq 3$.

~\newline
{\em Sufficiency:}
By Theorem \ref{thm:exact}, and due to the equivalence of the original Exact BVC and
$d$-Relaxed Exact BVC, for $d\geq 2$, $n\geq (d+1)f+1$ is sufficient for $d$-Relaxed Exact BVC.
Then by Lemma \ref{lemma:k:2}, this condition is also sufficient for $k$-Relaxed Exact BVC where $2\leq k\leq d-1$.

{\em Necessity:}
	We first prove that $n\geq d+2$ is necessary for $f=1$ and $2\leq k\leq d-1$.
 By Lemma \ref{lemma:k:1}, we only need to prove the necessity for $k=2$.
The proof is by contradition. Suppose that $n=d+1$ and $k$-Relaxed BVC with $k=2$ is achievable using
a certain algorithm.

	
Let us suppose that exactly one process is Byzantine faulty, but the faulty
process correctly follows any specified algorithm. Due to this restricted behavior,
it is possible for all the processes to correctly learn the input of all the other processes.
If we can show that $d+1$ processes are insufficient despite the above constraint on the faulty process,
then $d+1$ are insufficient when arbitrary behaviors are allowed for the faulty process.
Hereafter, we assume that all the processes follow the specified algorithm.

Let $Y$ denote the multiset of inputs at all the $d+1$ processes.
and $N$ denote the multiset of the inputs of the non-faulty processes,
The output must satisfy the $k$-Relaxed Validity condition.
Thus, the output must be in $H_k(N)$. However,
since the identity of the faulty process is unknown, every process is potentially faulty.
Thus, to satisfy the $k$-Relaxed Validity condition, the output chosen by the algorithm
must be 
$$\Psi(Y) = \bigcap_{T\subseteq Y, |T|=|Y|-f}~{H_k}(T).$$
Recall that $k=2$ and $f=1$ presently.
Also recall the definition of $H_k$ presented earlier.
Observe that
\begin{eqnarray*}
\Psi(Y) & = & \bigcap_{T\subseteq Y, |T|=|Y|-f}~{H_k}(T) \\
& = & \bigcap_{T\subseteq Y, |T|=|Y|-f} \left(\bigcap_{D\in \mathcal{D}_k}~g_D^{-1}(\mathcal{H}(g_D(T)))\right) \\
& = & \bigcap_{\substack{D\in \mathcal{D}_k \\ T \subseteq Y, |T|=|Y|-f}} g_D^{-1}(\mathcal{H}(g_D(T))) 
\end{eqnarray*}

To guarantee that the chosen output is in the convex hull of the inputs of non-faulty
processes, regardless of which process is faulty, the output must be contained in $\Psi(Y)$ defined above.

	Let the $i^{th}$ column of the following $d\times (d+1)$ matrix $S$ be an input vector of the $i^{th}$ process, where $0<\epsilon\leq\gamma$. We now show that these $d+1$ inputs lead to empty $\Psi(Y)$ when $k=2$.
	\begin{equation*}
	S=\begin{pmatrix}
	\gamma & 0 & \cdots & \cdots & 0 &-\gamma\\
	\epsilon& \gamma & 0 & \cdots & 0 &-\gamma\\
	\vdots&  \ddots & \ddots & \ddots & \vdots &\vdots\\
	\epsilon & \cdots & \epsilon & \gamma & 0 & -\gamma\\
	\epsilon & \cdots & \cdots & \epsilon & \gamma & -\gamma\\
	\end{pmatrix}
	\end{equation*}
In column $i$, $1\leq i\leq d$, the first $i-1$ elements equal 0, the $i$-th element equals $\gamma$, and the rest of the elements equal $\epsilon$. In column $d+1$, all elements are $-\gamma$.
Let $s_i$ denote the $i$-th column of matrix $S$ above, that is, the
input of $i$-th process.
Thus, $Y=\{s_i~|~1\leq i\leq d+1\}$.
We will consider several different combinations of $D$ and $T$ now:
\begin{itemize}
\item Observation 1: Consider $D=\{i,j\}\subseteq [1,d]$ and $T=Y-\{s_{d+1}\}$. Since the $i$-th coordinate of all the vectors in $T$ is non-negative, it follows that the $i$-th coordinate of all the vectors in $g_D^{-1}(\mathcal{H}(g_D(T)))$ for these $D,T$ must be non-negative. Therefore, $i$-th coordinate of all vectors in $\Psi(Y)$ must be non-negative. This holds for all $i$, $1\leq i\leq d$.

\item Observation 2: Consider $D=\{i,i+1\}$ where $1\leq i\leq d-1$, and $T=Y-\{s_{i+1}\}$. For each vector in $T$, its $i+1$-th coordinate is 
smaller than or equal to its $i$-th coordinate since $0<\epsilon\leq\gamma$. Therefore, the $i+1$-th coordinate of each vector in $g_D^{-1}(\mathcal{H}(g_D(T)))$ for these $D,T$ must be $\leq$ its $i$-th coordinate. Therefore, the $i+1$-th coordinate of each vector in $\Psi(Y)$ must be $\leq$ its $i$-th coordinate.

\item Observation 3: Consider $D=\{1,2\}$ and $T=Y-\{s_1\}$. For each vector in $T$, its first coordinate is 
non-positive. Therefore, the first coordinate of each vector in $g_D^{-1}(\mathcal{H}(g_D(T)))$ for these $D,T$ must be non-positive.
Therefore, the first coordinate of each vector in $\Psi(Y)$ must be non-positive.

\item Observation 4: Consider $D=\{d-1,d\}$ and $T=Y-\{s_{d+1}\}$. Since the last element of each vectors in $T$ is $\geq \epsilon$, the last coordinate of each vector in $g_D^{-1}(\mathcal{H}(g_D(T)))$ for these $D,T$ must be $\geq \epsilon$.
Therefore, the last coordinate of each vector in $\Psi(Y)$ must be $\geq \epsilon$.

\end{itemize}
Observations 1 and 3 together imply that the first element of the output of 2-Relaxed Exact BVC must be $0$.
This conclusion and observations 1 and 2 together imply that the $i$-th element of the output for $i\leq d$ must be $0$.
But this contradicts Observation 4 that the $d$-th element of the output of 2-Relaxed Exact BVC must be $\geq\epsilon>0$.


Therefore, we have proved $n=d+1$ is not sufficient for $f=1$, $k=2$.
By Lemma \ref{lemma:k:1}, $n= d+1$ is not sufficient for $k>2$ as well.
For $f>1$, we can then use the well-known simulation approach \cite{lamport1982byzantine}
 to show that $n=(d+1)f$ is not sufficient $k\geq 2$.
Therefore, $n \geq (d + 1)f + 1$ is necessary for $f \geq 1$ and $k\geq 2$, completing the proof.
	
\end{proof}

\subsection{Asynchronous Systems}

In this section, we prove the necessary and sufficient condition for \textsl{k-Relaxed Approximate BVC} in an asynchronous system, where $2\leq k\leq d-1$. 
For $k=1$, $n\geq 3f+1$ is necessary and sufficient,
and for $k=d$, $n\geq (d+2)f+1$ is necessary and sufficient.
Bounds for $d=1,2$ are included in the above results.

\begin{theorem}
\label{thm:k:async}
	$n\geq (d+2)f+1$ is necessary and sufficient for $2\leq k\leq d-1$ in k-Relaxed BVC in an asynchronous system.
\end{theorem}
The proof is provided in the Appendix \ref{app:b}.

\section{$\boldsymbol{(\delta, p)}$-Relaxed Byzantine Vector Consensus}

\subsection{Synchoronous Systems}

In this section, we prove the necessary and sufficient condition for \textsl{$(\delta, p)$-Relaxed Exact BVC} in a synchronous system.

\begin{theorem}
\label{thm:delta:sync}
	$n\geq \max(3f+1,(d+1)f+1)$ is necessary and sufficient for $(\delta, p)$-Relaxed Exact BVC in a synchronous system, where $0<\delta<\infty$ and $1\leq p$.
\end{theorem}

\begin{proof}
When $d=1$, the inputs are scalar, and all the $L_p$ norms are identical. For the case of $d=1$,
$(\delta,p)$-relaxed consensus is equivalent to a problem that was addressed in 
prior work \cite{fischer1990easy}. For this case, it was shown that
$n\geq 3f+1$ is necessary and sufficient.
Therefore, in the rest of the proof, we assume $d\geq2$.

	{\em Sufficiency:}
	By Theorem \ref{thm:exact}, and due to the equivalence of the original Exact BVC and
		$(0, p)$-Relaxed Exact BVC, for $d\geq 2$ and $1\leq p$, $n\geq (d+1)f+1$ is sufficient for $(0, p)$-Relaxed Exact BVC. Then by Lemma \ref{lemma:delta:2}, this condition is also sufficient for $(\delta, p)$-Relaxed BVC where $0<\delta<\infty$.
	
	{\em Necessity:}
	We first prove that $n\geq d+2$ is necessary for $f=1$ and $p=\infty$. The proof is by contradiction.
	Suppose that $n=d+1$ and $(\delta,\infty)$-Relaxed Exact BVC is achievable using a certain algorithm.
	
	Analogous to the proof of Theorem \ref{thm:k:sync}, we assume that any faulty process follows the algorithm correctly.
	Let the $i^{th}$ column of the following $d\times (d+1)$ matrix $S$ be an input vector of the $i^{th}$ process, where $x>2d\delta$.
	\begin{equation*}
	S=\begin{pmatrix}
	x & 0 & \cdots & \cdots & 0 & 0\\
	0& x & 0 & \cdots & 0 & 0\\
	\vdots&  \ddots & \ddots & \ddots & \vdots &\vdots\\
	0 & \cdots & 0 & x & 0 & 0\\
	0 & \cdots & \cdots & 0 & x & 0\\
	\end{pmatrix}
	\end{equation*}
	
For $1\leq i\leq d$,
the $i$-th coordinate of the $i$-th input is $x$, and the rest of the coordinates are 0.
The $d+1$-th input is all-0.
Let $Y$ denote the set of all inputs specified in matrix $S$.
If $N$ is the set of non-faulty processes, then the output must be in $H_{(\delta, \infty)}(N)$.
However, since the identity of any faulty process is unknown, the decision vector must be in 
	\begin{equation*}
	\bigcap_{T\subseteq Y, |T|=|Y|-f}H_{(\delta, \infty)}(T)
	\end{equation*}
where $f=1$.

Now we consider different choices of $T$:
\begin{itemize}
	\item 
	Observation 1: Consider $T$ as the set of all inputs except the input of process $i$, $1\leq i\leq d$.
	Then the $i^{th}$ element of each of the $d$ inputs in $T$ is $0$. Therefore then $i^{th}$ element of
	all the vectors in $H_{(\delta,\infty)}(T)$ -- and consequently in the output --  must be less or equal than $\delta$ due to the definition of $(\delta,\infty)$-Relaxed Validity. 
	
	\item
	Observation 2: Consider $T$ as the set of all inputs except the input of process $(d+1)$. The vectors
	in $H_{(\delta,\infty)}(T)$ are within distance $\delta$ (where the
distance is measured using the $L_\infty$ norm) of the convex hull $\mathcal{H}(T)$. In each convex combination of elements in $T$ used to obtain the convex hull $\mathcal{H}(T)$,
	at least one of the weight must be $\geq \frac{1}{d}$. Hence at least one element of each
	vector in $\mathcal{H}(T)$ must be $\geq \frac{x}{d}$.
	Thus, at least one element of each vector in $H_{(\delta,\infty)}(T)$ -- and consequently
	the output -- must be $\geq \frac{x}{d}-\delta>\delta$ (recall that $\frac{x}{d}-\delta>\delta$).
\end{itemize}

Thus, Observation 1 and 2 contradict each other, proving that $n=d+1$ is not sufficient for $f=1$. For $f>1$, we can use the simulation approach to show $n=(d+1)f$ is not sufficient \cite{lamport1982byzantine}. Therefore, $n \geq (d + 1)f + 1$ is necessary for \textsl{$(\delta, \infty)$-Relaxed Exact BVC} with $f>1$.


Now, for any vector $x$, $\|x\|_\infty\leq \|x\|_p$, for $1\leq p<\infty$ \cite{kothe1983topological}.
Therefore, we have
	$$
	H_{(\delta, p)}\subseteq H_{(\delta, \infty)}
	$$
	Then, the argument above for $(\delta,\infty)$-consensus would
imply that $n \geq (d + 1)f + 1$ is also necessary for \textsl{$(\delta, p)$-Relaxed Exact BVC}. 

\end{proof}

\subsection{Asynchoronous Systems}

In this section, we prove the necessary and sufficient condition for \textsl{$(\delta, p)$-Relaxed Approximate BVC} in an asynchronous system.

\begin{theorem}
\label{thm:delta:async}
	$n\geq (d+2)f+1$ is necessary and sufficient for $(\delta, p)$-Relaxed Approximate BVC in an asynchronous system, where $0<\delta<\infty$ and $1\leq p$.
\end{theorem}
The proof is provided in the Appendix \ref{app:c}.

\section{Relationship with Tverberg's Theorem}

\begin{theorem}
	(Tverberg's Theorem\cite{tverberg1966generalization}) For any integer $f\geq1$, and for every multiset $Y$ containing at least $(d+1)f+1$ points in $\mathbb{R}^d$, there exists a partition $Y_1,\cdots, Y_{f+1}$ of $Y$ into $f+1$ non-empty multisets such that $\bigcap_{l=1}^{f+1}\mathcal{H}(Y_l)\neq\emptyset$.
\end{theorem}

Since $\mathcal{H}(Y_l)\subseteq \mathcal{H}_k(Y_l)$ and 
$\mathcal{H}(Y_l)\subseteq {H}_{(\delta,p)}(Y_l)$,
it follows that Tverberg's theorem remains valid even if $\mathcal{H}(Y_l)$
is replaced in the statment of the theorem by
${H}_k(Y_l)$ or ${H}_{(\delta,p)}(Y_l)$.

The lower bound of $(d+1)f+1$ in Tverberg's theorem above is tight in the sense
that for $n\leq (d+1)f$, there exists a set of $n$ points such that for
every partition of the $n$ points into $f+1$ non-empty multisets
$Y_1,\cdots, Y_{f+1}$, $\bigcap_{l=1}^{f+1}\mathcal{H}(Y_l)=\emptyset$.  
Our impossibility results in the previous sections imply that the bound $(d+1)f+1$
remains tight even if we replace 
$\mathcal{H}(Y_l)$ by
${H}_k(Y_l)$ or ${H}_{(\delta,p)}(Y_l)$.
In particular, for $n\leq (d+1)f$, there exists a set of $n$ points such that for
every partition of the $n$ points into $f+1$ non-empty multisets
$Y_1,\cdots, Y_{f+1}$, $\bigcap_{l=1}^{f+1}{H}_k(Y_l)=\emptyset$.  
Similarly, for $n\leq (d+1)f$, there exists a set of $n$ points such that for
every partition of the $n$ points into $f+1$ non-empty multisets
$Y_1,\cdots, Y_{f+1}$, $\bigcap_{l=1}^{f+1}{H}_{(\delta,p)}(Y_l)=\emptyset$.

\section{Input-Dependent $\delta$ for $\boldsymbol{(\delta, p)}$-Relaxed Consensus in Synchronous Systems}

In the previous section, we showed that the tight necessary and sufficient condition for ${(\delta, p)}$-Relaxed Byzantine Vector Consensus
with constant $\delta>0$ is identical to that with $\delta=0$, when
$\delta$ is a constant. That is, the relaxation does not help. In this section, we show that for a {\em given} set of inputs, if we choose a relaxation parameter $\delta$ tha depends on the inputs themselves, then
the $(\delta, p)$-Relaxed Exact BVC problem is solvable with fewer than $(d+1)f+1$ processes, when $d>2$.

In particular, we define an input-dependent $\delta$ as follows. Let $v_i$ be the input at a non-faulty process $i$,
and let $I$ be the multiset of inputs at the non-faulty processes.
Define $E_+$ as the set of edges between the inputs at the non-faulty processes, where each input is viewed as
a point in the $d$-dimensional space.
Then, we require that input-dependent $\delta$ must be bounded as follows:
$$
\delta \leq \kappa(n,f,d,p) \max_{e\in E_+} \|e\|_p,
$$
where $\kappa(n,f,d,p)$ is a finite constant that may depend on number of processes $n$, number of failures $f$, dimension of the inputs $d$ and $L_p$ norm, but not on the inputs.
Intuitively, if the inputs at the non-faulty processes are far apart, then the above constraint allows the output to be farther away from the convex hull of the non-faulty inputs.

It is known that $n\geq 3f+1$ is the lower bound on the number of processes to achieve $(0,p)$-consensus \cite{lynch1989hundred}.
Similarly, we can also show the following.
\begin{lemma}
\label{lemma:imposs}
Input-dependent $(\delta,p)$-consensus is impossible with $n\leq 3f$.
\end{lemma}
The proof is similar to the proof in \cite{lynch1989hundred}, and is provided in the Appendix \ref{app:a}.

In the remaining discussion, we assume $n\geq 3f+1$. For
$d=1,2$, $n\geq 3f+1$ suffices to solve $(0,p)$-consensus. Thus, hereafter only $d\geq 3$ is interesting. However,
some of the claims below sometime hold for $d<3$ too, and therefore, sometimes we will also allow smaller values of $d$.

We will first
derive bounds on input-dependent $\delta$ for ${(\delta, 2)}$-Relaxed Byzantine Vector Consensus with $f=1$,
and then extend them to $(\delta,p)$-consensus with other values of $f$ and $p$.
We prove these results constructively, by showing that the algorithm presented below can
solve the problem
under a certain constraint on input-dependent $\delta$. The algorithm below is a modification of 
\textsl{Exact BVC algorithm} in \cite{vaidya2013byzantine} to incorporate $(\delta,p)$-relaxation.

\paragraph{Algorithm ALGO:}
\begin{itemize}
\item Step 1: Each process $i$ performs a Byzantine broadcast of its $d$-dimensional input $v_i$. Byzantine broadcast of each element of the vector $v_i$ can be 
performed separately by using any Byzantine broadcast algorithm, such as \cite{lamport1982byzantine}.
$n\geq 3f+1$ suffices for the correctness of Byzantine broadcast in a completely connected network.
At the completion of Step 1, each process will receive the multiset $S=\{a_i~|~1\leq i\leq n\}$, where
for a non-faulty process $i$, $a_i=v_i$, the input of process $i$, and for a faulty process $j$, $a_j$ may be any arbitrary point
in the $d$-dimensional space. 
Importantly, all non-faulty processes obtain identical set $S$.

The points in $S$ received from non-faulty processes are said to {\em non-faulty inputs},
and the remaining points are said to be {\em faulty inputs}.

\item Step 2: Each process
determines the smallest value $\delta$ such that
$\Gamma_{(\delta,2)}(S)=\bigcap_{T\subseteq S, |T|=|S|-f}H_{(\delta,2)}(T)$
is non-empty, and for this value of $\delta$, the process deterministically
chooses
a point in $\Gamma_{(\delta,2)}(S)$ as its output. All processes use identical
deterministic function to choose the output from $\Gamma_{(\delta,2)}(S)$.
\end{itemize}
Given set $S$,
let $\delta^*(S)$ denote the smallest value of $\delta$ for which 
$\Gamma_{(\delta,2)}(S)$ is non-empty. 
Although $\delta^*(S)$ will depend on $S$, the goal here is to determine a bound on
$\delta^*(S)$ that depends only on the inputs at the non-faulty processes.
In particular, we will first prove the claim below (in Theorem \ref{e_min_basic}).

\begin{itemize}
\item
Recall that the input at each process is a point in the $d$-dimensional Euclidean space. Let
$E_+$ denote the set of edges between the inputs at the non-faulty processes. Then,
$$\delta^*(S)< \min\left(\frac{\min_{e\in E_+}\|e\|_2}{2},~~ \frac{\max_{e\in E_{+}}\|e\|_2}{d-1}\right)$$
That is, $\delta^*(S)$ is upper bounded as shown above, and
${(\delta^*(S), 2)}$-Relaxed Byzantine Vector Consensus is achievable with  $n=d+1$, $f=1$ and $d\geq 3$.
Observe that although $\delta^*(S)$ will depend on the vectors in
$S$ corresponding to the faulty processes, the upper bound above does not depend
on those vectors.
\end{itemize}
Subsequently, we will obtain a bound for one case of $f\geq 2$, and propose a conjecture for the remaining case of $f\geq 2$.
In the rest of the discussion, we assume $n\geq 3f+1$, without necessarily stating
this explicitly again. The assumption of $n\geq 3f+1$ is necessary
for correctness of Byzantine broadcast used in the algorithm above.\footnote{When
the underlying network is a reliable broadcast channel, instead of
a point-to-point network, $n$ does not need to be exceed $3f$.
For such a broadcast network, the bounds for $d=2$ can be improved
similar to the bounds fo $d\geq 3$ derived in this paper.}

\subsection{Useful Lemmas}

In this section, we assume that $f=1$ and $n=d+1$.
Let the set of vectors obtained in Step 1 of algorithm ALGO be $S=\{a_1,a_2,\cdots,a_{d+1}\}$.
We consider the special case when the $d+1$ vectors in set $S$
are affinely independent.\footnote{When the vectors are not affinely independent, it is easy to show that $\delta=0$ can be
achieved, that is, $\delta^*=0$ (we will discuss this case later).}
Since the $d$-dimensional vectors $a_1,\cdots,a_{d+1}$ are affinely independent, they form a simplex.
Also, the vectors $b_i=\{a_i-a_{d+1}\}$, $1\leq i\leq d$, are linearly independent.

Let matrix $A=[a_1-a_{d+1},\cdots,a_d-a_{d+1}]$.
Consider matrix $B=[b_1,\cdots,b_{d}]$ such that
$B = (A^{-1})^T$. Define $b_{d+1}=-\sum_{i=1}^d b_i$. 
$\delta_{ij}$ is Kronecker's delta. Thus, $\delta_{ij}=1$ if and only if $i=j$, and 0 otherwise.
For $d$-dimensional vectors $a,b$,
$\langle a,b\rangle$ denotes their dot product.

The following two lemmas were proved by Akira \cite{akira2014radii}.

	\begin{lemma}\cite{akira2014radii}\label{akira2014radii_1}
		Let $a_1,\cdots,a_{d+1}$ and $b_1,\cdots,b_{d+1}$ be as defined above. Then, $\langle a_i-a_j,b_k\rangle=\delta_{ik}-\delta_{jk}$.
	\end{lemma}
	
	\begin{lemma}\cite{akira2014radii}\label{akira2014radii_2}
		Let $r$ be the radius of the inscribed sphere in the simplex formed by $a_1,\cdots,a_{d+1}$. Then,
 $r=\frac{1}{\sum_{i=1}^{d+1}\|b_i\|}$.
	\end{lemma}
	
	\begin{lemma}\label{delta_lemma}
		Let $r$ be the radius of the inscribed sphere of the simplex formed by
	points in $S=\{a_1,\cdots,a_{d+1}\}$.
	Then, $\delta^*(S)=r$.
	\end{lemma}
	
	\begin{proof}
In Step 2 of the algorithm above, recall that
$\Gamma_{(\delta,2)}(S)=\bigcap_{T\subseteq S, |T|=|S|-f}H_{(\delta,2)}(T)$.
Presently, $f=1$ and $S=\{a_1,a_2,\cdots,a_{d+1}\}$.
Thus, the convex hull of each subset $T\subset S$ such that $|T|=|S|-1=d$ is simply a facet 
of the simplex formed by $S=\{a_1,a_2,\cdots,a_{d+1}\}$.
Therefore, it follows that any point in $\Gamma_{(\delta,2)}(S)=\bigcap_{T\subseteq S, |T|=|S|-f}H_{(\delta,2)}(T)$
must be at distance at most $\delta$ from each facet of the simplex. Then, by the
definition of the inscribed sphere, $(r,2)$-consensus is achievable. Thus, $\delta^*(S)\leq r$.

Now suppose that $\delta^*<r$. This means there exists a point $p$ such that the distance from $p$ to all the facets of the simplex is less than $r$. This contradicts with the fact that $r$ is the radius of the inscribed sphere. Therefore $\delta^*(S)\not< r$. That is, $\delta^*(S)=r$.
	\end{proof}
	
	\begin{lemma}\label{radius_lemma}
		Assume $d\geq 2$. Let $r$ be the radius of the inscribed sphere of the simplex
formed by $S=\{a_1,\cdots,a_{d+1}\}$.
Let $\pi_k$ denote the facet of the simplex that contains $\{a_i~|~i\neq k,~1\leq i\leq d+1\}$ (i.e.,
all vertices except $a_k$),  $k=1,\cdots,d+1$.
Then $\pi_k$ itself is a simplex in a $(d-1)$-dimensionsional subspace.
Let $r_k$ be the radius of the $(d-1)$-dimensional inscribed sphere of $\pi_k$
in this $(d-1)$-dimensionsional subspace
 Then, $r<\min_{1\leq k\leq d+1}\,{r_k}$.
	\end{lemma}


	\begin{proof}

		By Lemma \ref{akira2014radii_2}, we can write $r=1/\sum_{i=1}^{d+1}\|b_i\|$.
		
		$\pi_k$ is the facet of the simplex that contains all vertices except $a_k$. From Lemma \ref{akira2014radii_1}, we know $b_k$ is orthogonal to $\pi_k$. In order to derive $r_k$, we need to determine the distance of a point $x$ in $\pi_k$ to the face $\pi_{jk}$ consisting of all points except $a_k$ and $a_j$. We first show that $b_{jk}$ defined as $b_{jk}=b_j-\frac{\langle b_j, b_k \rangle}{\|b_k\|^2}b_k$ is orthogonal to $\pi_{jk}$ and $b_k$.

		By Lemma \ref{akira2014radii_1}, for $\forall l,m$ such that $m\neq j$, $m\neq k$, $l\neq j$, $l\neq k$ 
		$$
		\langle b_{jk},a_m-a_l\rangle=\langle b_j-\frac{\langle b_j, b_k \rangle}{\|b_k\|^2}b_k,a_m-a_l\rangle=\delta_{jm}-\delta_{jl}-\frac{\langle b_j, b_k \rangle}{\|b_k\|^2}(\delta_{km}-\delta_{kl})=0.
		$$
	Also,	
		$$
		\langle b_{jk},b_k\rangle=\langle b_j-\frac{\langle b_j, b_k \rangle}{\|b_k\|^2}b_k,b_k\rangle=\langle b_j, b_k \rangle-\langle b_j, b_k \rangle=0.
		$$
		
		Hence $b_{jk}$ is normal to $\pi_{jk}$ and $b_k$. \\

Let $x$ be the center of the inscribed sphere of $\pi_k$ in the $(d-1)$-dimensional
subspace containing $\pi_k$.
Then, $x$ is equi-distant from all $\pi_{jk}$, $j\neq k$, $1\leq j\leq d+1$.
Since $x$ is in $\pi_k$, we have $x=\sum_{i=1,\cdots,d+1,i\neq k} t_ia_i$, where $\sum_{i=1,\cdots,d+1,i\neq k}t_i=1$, $t_i\geq 0$. Then we have, for $m\neq j$ and $m\neq k$,
		\begin{equation*}
		\begin{aligned}
		r_k ~=~dist(x,\pi_{jk})&=\frac{|\langle x-a_m, b_{jk} \rangle|}{\|b_{jk}\|}\\
		&=\frac{1}{\|b_{jk}\|}\left|\left\langle\left(\left(\sum_{i=1,\cdots,d+1,i\neq k} t_ia_i\right)-\left(\sum_{i=1,\cdots,d+1,i\neq k} t_i\right)a_m\right),b_{jk} \right\rangle\right|\\
		&=\frac{1}{\|b_{jk}\|}|\sum_{i=1,\cdots,d+1,i\neq k}t_i\langle a_i-a_m,b_{jk} \rangle|\\
		&=\frac{1}{\|b_{jk}\|}\left|\sum_{i=1,\cdots,d+1,i\neq k}t_i\left\langle a_i-a_m,b_j-\frac{\langle b_j, b_k \rangle}{\|b_k\|^2}b_k \right\rangle\right|\\
		&=\frac{|t_j|}{\|b_{jk}\|}\\
		&=\frac{t_j}{\|b_{jk}\|}
		\end{aligned}
		\end{equation*}
		
		Now, 
		\begin{eqnarray}
		\|b_{jk}\|&=&\|b_j-\frac{\langle b_j, b_k \rangle}{\|b_k\|^2}b_k\| \nonumber \\
		&=&(\langle b_j-\frac{\langle b_j, b_k \rangle}{\|b_k\|^2}b_k,b_j-\frac{\langle b_j, b_k \rangle}{\|b_k\|^2}b_k \rangle )^{1/2} \nonumber \\
		&=& (\|b_j\|^2-\frac{2\langle b_j, b_k \rangle^2}{\|b_k\|^2}+\frac{\langle b_j, b_k \rangle^2}{\|b_k\|^2})^{1/2} \nonumber \\
		&=&(\|b_j\|^2-\frac{\langle b_j, b_k \rangle^2}{\|b_k\|^2})^{1/2} \nonumber \\
		&\leq& \|b_j\| \label{e:b}
		\end{eqnarray}
		Since, $r_k=\frac{t_j}{\|b_{jk}\|}$ for all $j\neq k$, $1\leq j\leq d+1$, we have
		$$r_k=\frac{\sum_{j=1,\cdots,d+1,j\neq k}t_j}{\sum_{j=1,\cdots,d+1,j\neq k}\|b_{jk}\|}=\frac{1}{\sum_{j=1,\cdots,d+1,j\neq k}\|b_{jk}\|}
		$$
		Therefore
		$$\frac{1}{r_k}=\sum_{j=1,\cdots,d+1,j\neq k}\|b_{jk}\|\leq \sum_{j=1,\cdots,d+1,j\neq k}\|b_{j}\|<\sum_{j=1}^{d+1}\|b_{j}\|=\frac{1}{r}$$ 
		where the first inequality above follows from (\ref{e:b}), and the second inequality follows
		from the fact that $B$ is invertible, and thus $\|b_k\|$ is non-zero.
		Therefore, we obtain $r<r_k$. Since for every facet $\pi_k$, the above inequality holds, we have $r<\min_{1\leq k\leq d+1}\,{r_k}$,
 completing the proof.
	\end{proof}

	\begin{lemma}\label{e_max_lemma}
	Let $d\geq 1$.
		Let $r$ be the radius of the inscribed sphere of the simplex formed by
the points in  $S=\{a_1,\cdots,a_{d+1}\}$. We have $r<\frac{\max_{e\in E}\|e\|_2}{d}$, where $E$ is the set of all edges of the simplex.
	\end{lemma}

	\begin{proof}
Since $B=[b_1,\cdots,b_d]$ is a invertible matrix, we have $b_i\neq 0$, for $1\leq i\leq d$ (i.e.,
$b_i$ does not equal the vector with all $d$ elements 0). Also, since matrix $B$ is
invertible, the $d$ vectors $b_i$,
$1\leq i\leq d$, are linearly independent; thus, we also have $b_{d+1}=-\sum_{i=1}^{d}b_i\neq 0$.
		From Lemma \ref{akira2014radii_2}, we have 
		$$
		r=\frac{1}{\sum_{i=1}^{d+1}\|b_i\|}<\frac{1}{\sum_{i=1}^{d}\|b_i\|}
		$$
		since $\|b_{d+1}\|>0$.
		
		For $1\leq k\leq d$, we have
		$$
		\|a_k-a_{d+1}\|\cdot \|b_k\|\geq|\langle a_k-a_{d+1},b_k \rangle|=1
		$$
		$$
		\Rightarrow~~\|b_k\|\geq \frac{1}{\|a_k-a_{d+1}\|}.
		$$
		Therefore,
		$$
		r<\frac{1}{\sum_{i=1}^{d}\|b_i\|}\leq \frac{1}{\sum_{i=1}^{d}\frac{1}{\|a_i-a_{d+1}\|}}\leq\frac{\sum_{i=1}^{d}\|a_i-a_{d+1}\|}{d^2}
		$$
		Here we used the inequality $\frac{d}{\sum_{i=1}^{d}\frac{1}{x_i}}\leq\frac{\sum_{i=1}^{d}x_i}{d}$.
		Let $E$ denote the set of all edges between the vertices $a_1,a_2,\cdots,a_{d+1}$. Then, from the last inequality above,
		it follows that
		$$r<\frac{\max_{e\in E}\|e\|_2}{d}$$
	\end{proof}

	\subsection{$(\delta, 2)$-Relaxed Exact BVC}

In this section, we derive upper bounds for an achievable input-dependent $\delta$.
In particular, we derive one bound that applies to $f\geq 1$, and another bound that only applies
to $f=1$.
	
	The tight necessary and sufficient condition for solving $(\delta, 2)$-Relaxed Exact BVC problem in synchronous systems is $n\geq\max\{(d+1)f+1,3f+1\}$ for inputs of dimension $d$. Recall that it is impossible to solve $(\delta, p)$-Relaxed Exact BVC for $n\leq 3f$. Therefore, we only need to discuss the situation when the number of processes $n$ is $3f+1\leq n \leq (d+1)f$. Thus, the dimension of inputs must be $d\geq 3$.
	
	\subsubsection{$f=1$ case}
	Consider the multiset $S=\{a_1,\cdots,a_n\}$ collected at the end of Step 1 of algorithm ALGO
	presented earlier. 
	We first consider the case when the vectors in the set
	$\{a_i-a_n~|~i\neq n, 1\leq i\leq n\}$ are not linearly independent.

	\begin{theorem}
	Let $f=1$, $d\geq 3$ and $4\leq n\leq d+1$.  Consider the set of $n$ inputs $a_1,\cdots,a_{n}$ in $S$ obtained in Step 1 of algorithm ALGO, and suppose that the vectors in $\{a_i-a_n~|~i\neq n, 1\leq i\leq n\}$ are not linearly independent.  Then,
$(0, 2)$-consensus can be achieved, and  $\delta^*(S)=0$.
	\end{theorem}
	
	\begin{proof}
Note that $n\leq d+1$.
		Since the $n-1$ vectors in $\{a_i-a_n~|~i\neq n, 1\leq i\leq n\}$ are not independent of each other, $\{a_i-a_{n}\}$ belong to a $d^\prime$-dimensional subspace $W$, where $d^\prime<n-1$. Then we can find a projection $P$ from $d$-dimensions to $d'$-dimensions, while perserving the distances between the points in $S$. That is, $\|a_i-a_j\|_2 = \|Pa_i-Pa_j\|_2$,
$1\leq i,j\leq n$.
Since $f=1$ and $n>d'+1$, from previous results, we know that $(0, 2)$-consensus is
achievable.
Equivalently, $\delta^*(S)=0$.
	\end{proof}

Now we focus on the case when the vectors in the set $\{a_i-a_n~|~i\neq n, 1\leq i\leq n\}$ are linearly independent. 
Recall that we defined $\delta^*(S)$ for a given set $S$ such that
for all $\delta\geq\delta^*(S)$, $(\delta,2)$-consensus is achieved
when set $S$ is the set obtained Step 1 of algorithm ALGO.

	\begin{theorem}\label{e_min_basic}
	Let $f=1$, $d\geq 3$ and $4\leq n\leq d+1$.  Consider the set $S=\{a_1,\cdots,a_{n}\}$ obtained in Step 1 of algorithm ALGO, and suppose that the vectors in $\{a_i-a_n~|~i\neq n, 1\leq i\leq n\}$ are linearly independent.  Then $$\delta^*(S)<\frac{\min_{e\in E}\|e\|_2}{2}\leq \frac{\min_{e\in E_+}\|e\|_2}{2},$$
and
$$\delta^*(S)<\frac{\max_{e\in E_+}\|e\|_2}{n-2},$$
 where $E$ is the set of edges between any pair of inputs in $S$, and $E_+$ is the set of edges between any pair of non-faulty inputs in $S$.
	\end{theorem}
	
	\begin{proof}

We divide the proof into two cases: $n=d+1$ and $4\leq n<d+1$.

\paragraph{Case I: $n=d+1$:}
In this case, $n-1=d$.
		Since the $d$ vectors in $\{a_i-a_n~|~i\neq n, 1\leq i\leq n\}$ are linearly independent, we know that the inputs in $S$ form a simplex in $d$ dimensions. By Lemma \ref{delta_lemma}, $\delta(S)=r$ where $r$ is the radius of the inscribed sphere of the simplex. We can prove the theorem by induction on $d$.
		
		First consider $d=2$. When $d=2$, the simplex is simply a triangle.
Let the lengths of the three sides of the triangle be denoted as $a,b,c$, where $c\leq b\leq a$,
and define $p=\frac{a+b+c}{2}$.
$c$ must be positive, since no two inputs in $S$ are identical (otherwise,
the vectors in $\{a_i-a_n~|~i\neq n, 1\leq i\leq n\}$ will not be linearly independent).

By Heron's formula, 
the area of the triangle is given by $\sqrt{p(p-a)(p-b)(p-c)}$.  Then, the radius $r$ of the inscribed sphere (or
incircle, since $d=2$) is given by, 
		\begin{equation*}
		\begin{aligned}
		r&=\frac{\sqrt{p(p-a)(p-b)(p-c)}}{p}\\
		&=\sqrt{\frac{(p-a)(p-b)(p-c)}{p}}\\
		&\leq \frac{(p-a)+(p-b)}{2}\sqrt{\frac{p-c}{p}} \mbox{~~~~because $\frac{\alpha^2+\beta^2}{2}\geq \alpha\beta$}\\
		&<\frac{c}{2} \mbox{~~because $c>0$} \\
		&\leq \frac{\min_{e\in E}\|e\|_2}{2}\\
		\end{aligned}
		\end{equation*}
Thus, $r<\frac{\min_{e\in E}\|e\|_2}{2}$ when $d=2$.\\

Now, suppose that, for every simplex of dimension $k$,  $2\leq k$,
the radius of the inscribed sphere is less than half
the minimum distance between any two of its vertices.
Consider a simplex of dimension $k+1$.
Lemma \ref{radius_lemma} and the above assumption together imply that,
for a simplex in $k+1$ dimensions as well, the radius of the inscribed sphere is less than $\frac{\min_{e\in E}\|e\|_2}{2}$,
and therefore, also less than $\leq \frac{\min_{e\in E_+}\|e\|_2}{2}$. 
	
Now we prove that 
$$\delta(S)<\frac{\max_{e\in E_+}\|e\|_2}{d-1}$$.
	
Without loss of generality, assume that process 1 is faulty, and thus $a_1\in S$ is the
only faulty input $S$.
Recall that the $n=d+1$ points in $S$ form a simplex.
Let $\pi_1$ be the facet of the simplex formed by the points in $S-\{a_1\}$.
Observe that $\pi_1$ is isomorphic to a simplex in $d-1$ dimensions. Let $r_1$ be the radius of
$(d-1)$-dimensional inscribed sphere of $\pi_1$. By Lemma \ref{e_max_lemma}, we have $r_1<\frac{\max_{e\in E'}\|e\|_2}{d-1}$, where $E'$ is the set of edges between the inpue corresponding to $\pi_1$ (i.e., inputs in $S-\{p_1\}$). Since $\pi_1$ only contains non-faulty inputs, we have $E'=E_{+}$. By Lemma \ref{radius_lemma}, we have $r<r_1<\frac{\max_{e\in E_+}\|e\|_2}{d-1}$, completing the proof of Case I (recall that $d-1=n-2$
in this case).

\paragraph{Case II: $4\leq n<d+1$:}
	
		Since the vectors in $\{a_i-a_{n}~|~1\leq i<n\}$ are linearly independent, these vectors form a $n-1$ dimensional subspace $W$ (where $n-1<d$). Then we can find a projection matrix $P$ that projects these $d$-dimensional vectors into a $(n-1)$-dimensional space, while perserving the distances between the points in $S=\{a_1,\cdots,a_n\}$. Then the $n$ points $Pa_1,\cdots,Pa_{n}$ form a simplex in a $(n-1)$-dimensional subspace. By
the results in Case I, and substituting $d$ by $n-1$, the claim follows in Case II.
	\end{proof}

	\subsubsection{$f\geq 2$}

In this section, we focus on $f\geq 2$.
The proof for the case of $f\geq 2$ can potentially be adapted for $f\geq 1$. However,
we handled the case of $f=1$ in the previous section, because the proof for $f=1$ is simpler
than that for  $f\geq 2$. 
We first give a proof of bound for $n=(d+1)f$ inputs case, and then leave a conjecture for the remaining case.



	\begin{theorem}[Helly's theorem\cite{danzer1963helly}]\label{helly}
		Let $X_1,\cdots,X_n$ be a collection of compact convex subsets of $\mathbb{R}^d$, where $n\geq d+1$. If the intersection of every $d+1$ of these sets is nonempty, then
		$$
		\bigcap_{i=1}^{n} X_i\neq \emptyset
		$$
	\end{theorem}
	\begin{theorem}[Caratheodory's theorem\cite{barany1982generalization}]\label{caratheodory}
		$S$ is a set of points in $\mathbb{R}^d$. If $x\in \mathcal{H}(S)$, then $x\in \mathcal{H}(R)$ for some $R\subseteq S$, $|R|\leq d+1$.
	\end{theorem}
	
	\begin{theorem}\label{e_max_general}
	Let $f\geq 2$, $d\geq 3$ and $n=(d+1)f$.  Consider the set of $n$ inputs $S=\{a_1,\cdots,a_{n}\}$ obtained in Step 1 of algorithm ALGO. Then,
$$\delta^*(S)<\frac{\max_{e\in E_+}\|e\|_2}{d-1},$$
where $E_+$ is the set of edges between pairs of non-faulty inputs in $S$.

	\end{theorem}
	
	\begin{proof}
%
Consider multset $S$.
If $\max_{e\in E_+}\|e\|_2=0$, then the input of each non-faulty processes is identical to, say, $a_*$. Thus, at least $n-f=df$ points in $S$ equal $a_*$.
Thus, one subset of $S$ of size $n-f$ contains only $a_*$.
 Also, since $df\not<f+1$, each
subset of $S$ of size $n-f$ contains $a_*$.
Then, in Step 2 of algorithm ALGO, for $\delta=0$
$\Gamma_{(\delta,2)}(S)=\{a_*\}$.
Thus, each non-faulty process will choose $a_*$ as its output,
achieving $(0,2)$-consensus.
%
%

Hereafter, let us assume that $\max_{e\in E_+}\|e\|_2>0$.
We want to derive an upper bound on $\delta^*(S)$ such that $\Gamma_{(\delta^*(S),2)}(S)=\bigcap_{T\subseteq S, |T|=|S|-f}H_{(\delta^*(S),2)}(T)$ is not empty.

In the following, for brevity, we may refer to $\delta^*(S)$ simply as $\delta$.

 Let $P_i$, $i=1,\cdots,\binom{n}{f}$ be the subsets of $S$ of size $(n-f)=df$, and let $F_i=S-P_i$, $i=1,\cdots,\binom{n}{f}$. Now, since any $f$ of the processes may be faulty, just one of these size $n-f$ subsets ($P_i$'s)
is guaranteed to contain only non-faulty inputs. Therefore, we obtain the following equation for $\delta$,
by observing that the output of $(\delta,p)$- consensus must not be farther than $\delta$
from the convex hull of the non-faulty inputs.

		$$
		\delta=\min_{p\in \mathbb{R}^d} \max_{i=1,\cdots,\binom{n}{f}} dist(p,\mathcal{H}(P_i))
		$$

Let $p_0 \in \arg (\min_{p\in \mathbb{R}^d} \max_{i=1,\cdots,\binom{n}{f}} dist(p,\mathcal{H}(P_i)))$.


		Let $\{P_i\}$ be the set containing all $P_i$ where $i=1,\cdots,\binom{n}{f}$. Let $Q_1,\cdots,Q_m$ be all
the distinct subsets of $S$ such that $Q_i=P_{j_i}$ for some $j_i$ and $dist(p_0,\mathcal{H}(Q_i))=\delta$, $1\leq i\leq m$.


		Now we consider the following two cases.
		
		\begin{itemize}
			\item Case 1: $1\leq m\leq d$:
			
			Consider the intersection of $Q_i$. Since $|Q_i|=df$ and $|S-Q_i|=f$,
			$$
			\Bigg|\bigcap_{i=1}^m Q_i\Bigg|\geq (d+1)f-mf>0
			$$
			This implies that
\begin{eqnarray}
\bigcap_{i=1}^m \mathcal{H}(Q_i)\neq \emptyset. \label{e_HQi}
\end{eqnarray}
			
In Case 1, by contradiction, we prove that $\delta=0$.


 Suppose that $\delta>0$. Let $x_i\in Q_i$ such that $dist(p_0,x_i)=dist(p_0,\mathcal{H}(Q_i))=\delta$ where $i=1,\cdots,m$. Let $\pi^i$ be the supporting hyperplane of $Q_i$ such that $x_i\in \pi^i$ and $dist(p_0,\pi^i)=\delta$, $i=1,\cdots,m$.
%
%
 Let $S^i_+$ denote the half-space that contains $p_0$ and is delimited by $\pi^i$, and let $S^i_-$ denote the other half-space delimited by $\pi^i$, which contains $\mathcal{H}(Q_i)$. Let $S^i_+$ be the interior of $\pi^i$, and let $n^i$ denote the inward-pointing normal vector of $\pi^i$. We have $\bigcap_{i=1}^m S^i_-\neq\emptyset$, otherwise $\bigcap_{i=1}^m \mathcal{H}(Q_i)= \emptyset$, contradicting (\ref{e_HQi}).

Since $\delta>0$, $p_0$ is in the interior of $S_+^i$, $1\leq i\leq m$.
 Let $x\in \bigcap_{i=1}^m S^i_-$. Then, for the unit vector $q=\frac{p_0-x}{\|p_0-x\|_2}$ 
			$$
				\langle q, n^i\rangle > 0,   ~~~~~1\leq i\leq m
			$$
			since $p_0$ is in the interior of $S^i_+$, $x\in S^i_-$, and $n^i$ is the normal vector pointing towards $p_0$ of hyperplane $\pi^i$.
			
			By definition of $p_0$, we have $dist(p_0,\mathcal{H}(P_i))<\delta$ for $P_i\not\in \{Q_j~|~1\leq j\leq m\}$. Define $\epsilon$ such that
			$$\epsilon=\delta-\max_{P_i\not\in \{Q_j\}} \, dist(p_0,\mathcal{H}(P_i))$$
			Thus, $\delta\geq \epsilon>0$. Now we consider the distance of point $p'=p_0-\alpha q$ to all $P_i$'s, where 
			$$
				0<\alpha\ < \min \left(\min_{1\leq i\leq m}2\langle p_0-x_i,q\rangle,~\epsilon\right)
			$$
			Now, for $1\leq i\leq m$,
			\begin{equation*}
			\begin{aligned}
				dist(p',\mathcal{H}(Q_i))
				&\leq dist(p',x_i) \mbox{~~~~because $x_i\in \mathcal{H}(Q_i)$}\\
				&=(\langle p_0-\alpha q-x_i, p_0-\alpha q-x_i\rangle)^\frac{1}{2} \mbox{~~~~because $p'=p_0-\alpha q$}\\
				&=(\langle p_0-x_i, p_0-x_i\rangle-2\langle p_0-x_i, \alpha q\rangle+\alpha^2)^{\frac{1}{2}}\mbox{~~~~because $q$ is a unit vector}\\
				&=(\delta^2+\alpha(\alpha-2\langle p_0-x_i,  q\rangle))^\frac{1}{2}\\
				&<\delta \mbox{~~~~because of the definition of $\alpha$}
			\end{aligned}
			\end{equation*}
			For $P_i\not\in\{Q_j\}$, we have
			\begin{equation*}
			\begin{aligned}
				dist(p',\mathcal{H}(P_i))
				&\leq dist(p',p_0)+ dist(p_0,\mathcal{H}(P_i)) \mbox{~~~~by triangular inequality}\\
				&\leq\alpha + (\delta-\epsilon) \mbox{~~~~due to definitions of $p'$ and $\epsilon$}\\
				&<\delta \mbox{~~~~because by definition of $\alpha$, $\alpha<\epsilon$}
			\end{aligned}
			\end{equation*}
			Therefore, there exists a point $p'$ whose maximum distance to any $\mathcal{H}(P_i)$,
			$i=1,\cdots,\binom{n}{f}$,
			 is less than $\delta$, which contradicts the definition of $\delta$. Therefore, we must have $\delta=0$ in Case 1.
			Thus, the theorem is trivially true in Case 1.

%
			
			\item Case 2: $m\geq d+1$:
		If $\delta=0$, then the theorem is trivially true.

Now suppose that $\delta>0$.  Let $\{Q_j\}$ denote the set $\{Q_j~|~j=1,\cdots,m\}$.
If the intersection of the convex hulls of every choice of $d+1$ sets in $\{Q_j\}$ is non-empty, by
Theorem \ref{helly}, it follows that the intersection of the convex hulls of {\em all} the $m$ sets
in $\{Q_j\}$ is non-empty.
Then, by an argument similar to Case 1 above, we can show that $\delta=0$, which contradicts with the assumption that $\delta>0$. 
Therefore, there must exist $d+1$ sets in $\{Q_j\}$ such that the intersection of their convex hulls is empty.

			Let $Q'_1,\cdots,Q'_{d+1}$ denote $d+1$ distinct sets in $\{Q_j\}$ such that $\bigcap_{i=1}^{d+1}\mathcal{H}(Q'_i)=\emptyset$. It follows that $\bigcap_{i=1}^{d+1}Q'_i=\emptyset$. Let $F'_i=S-Q_i'$, $i=1,\cdots,d+1$.

Now, $|\bigcap_{i=1}^dQ'_{i}|\geq (d+1-d)f=f$. Since $\bigcap_{i=1}^{d+1}Q'_{i}=\emptyset$, we have $\bigcap_{i=1}^d Q'_{i}\subseteq F'_{d+1}$, then $|\bigcap_{i=1}^dQ'_{i}|\leq |F'_{d+1}|=f$. Hence $|\bigcap_{i=1}^dQ'_{i}|=f$, and $\bigcap_{i=1}^dQ'_{i}=F'_{d+1}$. Similarly,
we can show that 
			$$\bigcap_{i=1,\cdots,d+1,i\neq k}Q'_{i}=F'_{k}$$
			We can also show that $F'_i$'s are disjoint. For $s\neq t$,
			$$F'_s\bigcap F'_t=\left(\bigcap_{i=1,\cdots,d+1,i\neq s}Q'_{i}\right)~~\bigcap~~ \left(\bigcap_{i=1,\cdots,d+1,i\neq t}Q'_{i}\right)=\bigcap_{i=1}^{d+1}Q'_{i}=\emptyset$$.
			
			Also, since $|F'_i|=f$ and $F'_s\bigcap F'_t=\emptyset$ for $s\neq t$, it
follows that $\bigcup_{i=1}^{d+1}F'_{i}=S$, and thus, $|\bigcup_{i=1}^{d+1}F'_{i}|=(d+1)f$.
Thus the $(d+1)$
$F_i'$'s form a partition of $S$.
 Finally, since $Q'_k\bigcap F'_k=\emptyset$,  $Q'_k=\bigcup_{i=1,\cdots,{d+1}, i\neq k}F'_i$.

\begin{claim}
\label{claim:simplex}
Consider a set $Z$ of size $d+1$ consisting of one point each in $F_i'$.
Then the $d+1$ points in $Z$ are affinely independent, and $\mathcal{H}(Z)$
is a simplex in $d$-dimensions.
\end{claim}
\begin{claimproof}
The proof is by contradiction.
Suppose that the $d+1$ points in $Z$ are not affinely independent.
Then there must exist a subspace of dimension $\leq d-1$ that contains the $d+1$ points in $Z$.
Let $Z\cap F_i'=\{z_i\}$.
Let $K_k=\mathcal{H}(Z-\{z_k\})$. Then, $z_k\in \bigcap_{i\neq k} K_i$. Thus, every $d$ of $K_i$'s have a non-empty intersection. Then, by Theorem \ref{helly}, we have $\bigcap_{i=1}^{d+1} K_i\neq \emptyset$. By definition of $K_i$, $K_i\subseteq \mathcal{H}(Q'_i)$, therefore $\bigcap_{i=1}^{d+1}\mathcal{H}(Q'_i)\neq \emptyset$, which contradicts with the fact that $\bigcap_{i=1}^{d+1}\mathcal{H}(Q'_i)=\emptyset$. This proves the lemma.
\end{claimproof}
			
Consider a point $w_i\in F'_i$, $i=1,\cdots, d+1$.
Let $W = \{w_i~|~1\leq i\leq d+1\}$
and $W_k = W-\{w_k\}$, $1\leq k\leq d+1$.
Claim \ref{claim:simplex} implies that $\mathcal{H}(W)$ is a simplex.  Let us call this simplex $A$.

Consider the point $p_0$ defined previously as
$$
p_0\in\arg (\min_{p\in \mathbb{R}^d} \max_{i=1,\cdots,\binom{n}{f}} dist(p,\mathcal{H}(P_i)))
$$


\begin{claim}
\label{claim:Hw_i}
$\mathcal{H}(S)-\bigcup_{i=1}^{d+1}\mathcal{H}(Q'_i) \subseteq \mathcal{H}(W)$.
\end{claim}
\begin{claimproof}

 Consider any $x\in \mathcal{H}(S)-\bigcup_{i=1}^{d+1}\mathcal{H}(Q'_i)$. By Theorem \ref{caratheodory},
there exist $d+1$ points $v_1,v_2,\cdots,v_{d+1}$ such that
$\{v_1,v_2,\cdots,v_{d+1}\}\subset S$
and $x\in \mathcal{H}(\{v_1,v_2,\cdots,v_{d+1}\})$.
Also, since $x\in \mathcal{H}(S)-\bigcup_{i=1}^{d+1}\mathcal{H}(Q'_i)$, $x\not\in \bigcup_{i=1}^{d+1}\mathcal{H}(Q'_i)$.

Recall that $\bigcup_{1\leq i\leq d+1}F_i'=S$.
Since $v_i\in S$ and $v_j\in S$,
suppose that $v_i\in F'_{i"}$ and $v_j \in F'_{j"}$.
We claim that for $i\neq j$, $F'_{i"}\neq F'_{j"}$.
Otherwise, there exists $F_k'$ such that $v_l\not\in F_k'$ for $1\leq l\leq d+1$.
Then $x\in \mathcal{H}(S-F'_k)=\mathcal{H}(Q'_k)$, contradicting the fact that $x\not\in \bigcup_{i=1}^{d+1}\mathcal{H}(Q'_i)$. Therefore, without loss of generality, let us assume that $v_i\in F'_i$ for $1\leq i\leq d+1$.

 Consider a sequence of convex hulls $V_0=\mathcal{H}(\{v_1,v_2,\cdots,v_{d+1}\})$, $V_1=\mathcal{H}(\{w_1,v_2,\cdots,v_{d+1}\})$, $V_2=\mathcal{H}(\{w_1,w_2,\cdots,v_{d+1}\})$, $\cdots$, $V_{d+1}=\mathcal{H}(\{w_1,w_2,\cdots,w_{d+1}\})$.
By definition of $V_0$, $x\in V_0$. We now show that if $x\in V_i$, $0\leq i<d+1$, then $x\in V_{i+1}$. The proof is by contradiction. Suppose that $$x\in V_i$$ and $$x\not\in V_{i+1}$$ for some $i$, $0\leq i< d+1$.

$x\in V_i$ and $x\not\in \bigcup_{i=1}^{d+1}\mathcal{H}(Q'_i)$ together imply
the existence of weights $\alpha_i$ such that
			$$
				x=\alpha_1 w_1+\cdots+\alpha_i w_i+\alpha_{i+1} v_{i+1}+\alpha_{i+2} v_{i+2}+\cdots+\alpha_{d+1} v_{d+1}
			$$
			with $\sum_{j=1}^{d+1}\alpha_j=1$, and $\alpha_j> 0$ for $1\leq j\leq d+1$.
			If any of the above weights were to be 0, then $x$ would be a convex combination
			of $d$ points in $S$, implying that $x$ is contained in the convex hull of one of the
			sets in $\{Q'_j~|~1\leq j\leq d+1\}$; this contradicts with $x\not\in \bigcup_{i=1}^{d+1}\mathcal{H}(Q'_i)$.

Now recall that $\{v_i,w_i\}\subseteq F'_i$, $1\leq i\leq d+1$. Therefore,
by Claim \ref{claim:simplex}, the points in the set
$\{w_1,\cdots,w_{i+1},v_{i+2},\cdots,v_{d+1}\}$ are affinely independent.
This observation together with 
$x\not\in V_{i+1}$ implies that there exist weights $\beta_i$'s such that
			$$
				x=\beta_1 w_1+\cdots+\beta_i w_i+\beta_{i+1} w_{i+1}+\beta_{i+2} v_{i+2}+\cdots+\beta_{d+1} v_{d+1}
			$$
			with $\sum_{j=1}^{d+1}\beta_j=1$, and at least one weight $\beta_j<0$.
(Note that if all the weights were to be non-negative, then we would have $x\in V_{i+1}$,
which would contradict the assumptiom above.)

 We can also show that $\beta_{i+1}>0$. The proof is by contradiction.
Suppose that $\beta_{i+1}\leq 0$. Then by simple arrangement of the above two expressions for $x$, we have
			\begin{eqnarray}
				&& \frac{1}{\alpha_{i+1}-\beta_{i+1}}(\alpha_{i+1} v_{i+1}-\beta_{i+1} w_{i+1}) \nonumber \\
				& =& ~-\frac{1}{\alpha_{i+1}-\beta_{i+1}}(\alpha_1 w_1+\cdots+\alpha_i w_i+\alpha_{i+2} v_{i+2}+\cdots+\alpha_{d+1} v_{d+1}) \nonumber \\
				&&+\frac{1}{\alpha_{i+1}-\beta_{i+1}}(\beta_1 w_1+\cdots+\beta_i w_i+\beta_{i+2} v_{i+2}+\cdots+\beta_{d+1} v_{d+1}) \label{eq:y}
			\end{eqnarray}
			Let $\pi_{i+1}$ denote the hyperplane passing through $d$ points $w_{1},\cdots, w_{i},v_{i+2},\cdots,v_{d+1}$. Let $y=\frac{1}{\alpha_{i+1}-\beta_{i+1}}(\alpha_{i+1} v_{i+1}-\beta_{i+1} w_{i+1})$; then $y\in \mathcal{H}(F'_{i+1})$
because $\{v_{i+1},w_{i+1}\}\subseteq F'_{i+1}$.


The left side of (\ref{eq:y}) equals $y$. Then the
right side of (\ref{eq:y}) implies that
 $y\in \pi_{i+1}$. Then we have $d+1$ points $w_{1},\cdots, w_{i}, y,v_{i+2},\cdots,v_{d+1}$ on a hyperplane $\pi_{i+1}$; denote these points as $u_1,\cdots,u_{d+1}$, respectively.
%
%
Define $H_k=\mathcal{H}(\{u_i~|~i\neq k,~1\leq i\leq d+1\}$.
Define $H = \{H_k~|~1\leq k\leq d+1\}$, which contains $d+1$ convex sets ($H_k$'s).
Observe that $u_{k'} \in  \bigcap_{k\neq k'} H_k$.
Thus, any $d$ convex sets in $H$ have a non-empty intersection.
Also, $H_k \subset \pi_{i+1}$, which is a $(d-1)$-dimensional subspace. Then
by Theorem \ref{helly}, we have $\bigcap_{k=1}^{d+1}H_k\neq \emptyset$. Since $u_i\in \mathcal{H}(F'_i)$ and $Q'_k=\bigcup_{i=1,\cdots,{d+1}, i\neq k}F'_i$, we have 
			$$H_k\subseteq \mathcal{H}\left( \bigcup_{i=1,\cdots,{d+1}, i\neq k}\mathcal{H}(F'_i)\right) = \mathcal{H}(Q'_k)$$
			Then, $\bigcap_{k=1}^{d+1}H_k\neq \emptyset$ implies that $\bigcap_{k=1}^{d+1}\mathcal{H}(Q'_k)\neq\emptyset$, which is a contradiction.
 Thus, we have proved that $\beta_{i+1}>0$.
			
			Define $k$ as an index such that $\beta_k<0$ and $\frac{\alpha_k}{|\beta_k|}=\min\limits_{\{l|\beta_i<0\}}(\frac{\alpha_l}{|\beta_l|})$. Since $\beta_{i+1}>0$, $k\neq i+1$.
Now by using the above two equations for $x$, we obtain,
			\begin{equation*}
			\begin{aligned}
				x=&\frac{1}{1+\alpha_k/|\beta_k|}\left((\alpha_1 w_1+\cdots+\alpha_i w_i+\alpha_{i+1} v_{i+1}+\alpha_{i+2} v_{i+2}+\cdots+\alpha_{d+1} v_{d+1})\right.\\
				&~~~~~~~~~~~~~~~~\left.+\alpha_k/|\beta_k|(\beta_1 w_1+\cdots+\beta_i w_i+\beta_{i+1} w_{i+1}+\beta_{i+2} v_{i+2}+\cdots+\beta_{d+1} v_{d+1})\right)\\
				=&\frac{1}{1+\alpha_k/|\beta_k|}\big(\sum_{j=1}^i(\alpha_j+\frac{\alpha_k}{|\beta_k|}\beta_j) w_j+ (\alpha_{i+1} v_{i+1}+\frac{\alpha_k}{|\beta_k|}\beta_{i+1} w_{i+1})
\sum_{j=i+2}^{d+1}(\alpha_j+\frac{\alpha_k}{|\beta_k|}\beta_j) v_j
\big)
			\end{aligned}
			\end{equation*}

Observe that the last summation represents a convex combination of
$w_j,~1\leq j\leq i+1$ and $v_j,i+1\leq j\leq d+1$.
The weights for each of these terms is non-negative, with the weight of the term
with index $k$ being equal to $\alpha_k+\frac{\alpha_k}{|\beta_k|}\beta_k=0$.
Since $\{v_k,w_k\}\subseteq F'_k$ and $Q'_k = S - F'_k$, we have that
$x\in \mathcal{H}(S-F'_k)=\mathcal{H}(Q'_k)$. This contradicts with the fact that $x\not\in \bigcup_{i=1}^{d+1}\mathcal{H}(Q'_i)$. Therefore, we have proved that $x\in V_{i+1}$.

By induction, we have $x\in V_{d+1}=\mathcal{H}(\{w_1,w_2,\cdots,w_{d+1}\})$. Hence we have 
			$$
				\mathcal{H}(S)-\bigcup_{i=1}^{d+1}\mathcal{H}(Q'_i)\subseteq \mathcal{H}(\{w_1,w_2,\cdots,w_{d+1}\})
			$$
			                                                
\end{claimproof}

\begin{claim}
$p_0$ is contained in the simplex $A$ formed by $W=\{w_i~|~1\leq i\leq d+1\}$,
i.e., $p_0\in\mathcal{H}(W)$.
\end{claim}
\begin{claimproof}
			
We first show that $p_0$ cannot be outside $\mathcal{H}(S)$. The proof is by contradiction.

Suppose that $p_0\not\in \mathcal{H}(S)$. Consider the distance $D=dist(p_0,\mathcal{H}(S))$. Since $\mathcal{H}(Q'_i)\subseteq \mathcal{H}(\bigcup_{i=1}^{d+1}Q'_i)=\mathcal{H}(S)$, we have $\delta\geq D$. Let $x$ denote the projection of $p_0$ on $\mathcal{H}(S)$, that is, $x\in \mathcal{H}(S)$ and $dist(p_0,x)=D$. Since $\mathcal{H}(S)$ is convex, we know that $x$ is unique. By Theorem \ref{caratheodory}, there exists a subset $S'\subseteq S$, such that $x\in \mathcal{H}(S')$ and $|S'|\leq d+1$.
Let us name $\mathcal{H}(S')$ as $B$.
 Consider the inputs in $S'$, there are two cases:
\begin{itemize}
\item (Case i) $F'_i\cap S' \neq\emptyset$ for $1\leq i\leq d+1$:
Since $S'$ is of size $d+1$, it follows that $|F'_i\cap S'|=1$.
Let $F'_i\cap S'=\{u_i\}$.
Then $S'=\{u_i~|~1\leq i\leq d+1\}$.
By Claim \ref{claim:simplex}, $\mathcal{H}(S')=B$ is a simplex.
Thus, $x$ is contained in the simplex $\mathcal{H}(S')=\mathcal{H}(\{u_i~|~1\leq i\leq d+1\})$.

\begin{itemize}
\item Recall that $D=dist(p_0,\mathcal{H}(S))=dist(p_0,x)$.
\item Since $x\in B$, by definition of $dist(p_0,B)$, $dist(p_0,x) \geq dist(p_0,B)$.
\item
For simplex $B$,
let $\pi^i$ denote the facet containing points in $S'-\{u_i\}$. Since $u_i\in F'_i$ and $Q'_i=\bigcup_{l=1,\cdots,{d+1}, l\neq i}F'_l$, we have $\pi^i\subseteq \mathcal{H}(Q'_i)$.

Since $p_0\not\in\mathcal{H}(S)$, and $S'\subset S$, it follows that $p_0\not\in \mathcal{H}(S')=B$.
Thus, projection $y$ of $p_0$ on simplex $B$ must be on some facet of $B$. 
Suppose that this facet of $B$ is $\pi^k$ (i.e., $y\in \pi^k$).
Thus, $dist(p_0,B)=dist(p_0,\pi^k)$.

\item Since $\pi^k \subseteq \mathcal{H}(Q'_k)$, $dist(p_0,\pi^k)\geq dist(p_0,\mathcal{H}(Q'_k))$.
\item Finally, by definition of $Q'_k$, $dist(p_0,\mathcal{H}(Q'_k))=\delta$.
\end{itemize}
The above five observations together imply that $D\geq\delta$.
We previously showed that $\delta\geq D$. Therefore,
$D=\delta$.
That is, $dist(p_0,\mathcal{H}(S))=\delta$.

Since $dist(p_0, \mathcal{H}(Q'_i))$ also equals $\delta$ for $i=1,\cdots, d+1$, and 
$Q'_i\subset S$, projection of $p_0$ on $\mathcal{H}(Q'_i)$ and projection of $p_0$ on $\mathcal{H}(S)$ must
be identical.\footnote{Since $\mathcal{H}(S)$ is convex, and $p_0\not\in\mathcal{S}$,
there is a unique point $x\in \mathcal{H}(S)$ that is at distance $D=\delta$ from $p_0$.}
%
%
This implies that $x\in \mathcal{H}(Q'_i)$ for $i=1,\cdots,d+1$, contradicting with the fact that $\bigcap_{i=1}^{d+1} \mathcal{H}(Q'_i)=\emptyset$.

\item (Case ii) There exists $k$, $1\leq k\leq d+1$ such that $F'_i\cap S'=\emptyset$. This, together
with the facts that $Q'_k=\bigcup_{i=1,\cdots,{d+1}, i\neq k}F'_i$, and $S'\subseteq S=\bigcup_{i=1,\cdots,{d+1}}F'_i$, implies that $S'\subseteq Q'_k$.
Hence $dist(p_0,B)=dist(p_0,\mathcal{H}(S'))\geq dist(p_0, \mathcal{H}(Q'_k))=\delta$. 
Similar to Case i, here too we have $D=dist(p_0,\mathcal{H}(S))=dist(p_0,x)\geq dist(p_0,B)$.
Therefore, $D\geq \delta$.

Since we already showed that $\delta\geq D$, we have $D=\delta$. Then, by similar argument
as Case i above, we can show that $x\in \mathcal{H}(Q'_i)$ for $i=1,\cdots,d+1$, contradicting with the fact that $\bigcap_{i=1}^{d+1} \mathcal{H}(Q'_i)=\emptyset$.
\end{itemize}
Therefore $p_0$ cannot be outside $\mathcal{H}(S)$. Thus, $p_0\in\mathcal{H}(S)$.

			We now show that $p_0\not\in\bigcup_{i=1}^{d+1}\mathcal{H}(Q'_i)$. By assumption, $dist(p_0,\mathcal{H}(Q'_i))=\delta>0$, for $1\leq i\leq d+1$. If $p_0\in \bigcup_{i=1}^{d+1}\mathcal{H}(Q'_i)$,
then there exists $k$ such that $p_0\in \mathcal{H}(Q'_k)$. Therefore, $dist(p_0,\mathcal{H}(Q'_k))=0$, which contradicts with the assumption that $\delta>0$.
			
			Thus, we have shown that $p_0\in \mathcal{H}(S)-\bigcup_{i=1}^{d+1}\mathcal{H}(Q'_i)$. Recall that simplex $A=\mathcal{H}(W)$. By Claim \ref{claim:Hw_i}, $\mathcal{H}(S)-\bigcup_{i=1}^{d+1}\mathcal{H}(Q'_i)\subseteq \mathcal{H}(W)=A$. Therefore, $p_0$ is in the simplex
$A=\mathcal{H}(W)$.

\end{claimproof}


			Let $\pi'_k$ denotes the facet of simplex $A$ that contains $\{w_i~|~i\neq k,~1\leq i\leq d+1\}$.
That is, $\pi'_k=\mathcal{H}(\{w_i~|~i\neq k,~1\leq i\leq d+1\}$.
 Since $w_i\in F'_i$ and $Q'_k=\bigcup_{i=1,\cdots,{d+1}, i\neq k}F'_i$, we have $\pi'_k\subseteq \mathcal{H}(Q'_k)$. Hence 
			$$
			\delta=dist(p_0,\mathcal{H}(Q'_i))\leq dist(p_0,\pi'_i)
			$$
			for $i=1,\cdots, d+1$.

Let $S_i$ denote the area (i.e., $(d-1)$-dimensional volume) of facet $\pi_i'$
of simplex $A$. Also, let $r_A$ be the
radius of the sphere inscribed in simplex $A$.
Then volume of simplex $A$ is given by
$\frac{1}{d} \sum_{i=1}^{d+1} S_i r_A$ because the center of the
inscribed sphere is at distance $r_A$ from all the facets of $A$.
Similarly, since $p_0$ is inside simplex $A$, the volume of $A$
is also given by $\frac{1}{d} \sum_{i=1}^{d+1} S_i \, dist(p_0,\pi'_i)$.
Since $\pi_i'\subseteq \mathcal{H}(Q_i')$ and $dist(p_0,\mathcal{H}(Q_i'))=\delta$,
we have $\delta=dist(p_0,\mathcal{H}(Q_i'))\leq dist(p_0,\pi'_i)$ for $1\leq i\leq d+1$.
Thus, we get
\begin{eqnarray}
&& ~~~~~ \frac{1}{d} \sum_{i=1}^{d+1} S_i r_A ~=~
\frac{1}{d} \sum_{i=1}^{d+1} S_i \, dist(p_0,\pi'_i) ~\geq
\frac{1}{d} \sum_{i=1}^{d+1} S_i \, \delta \nonumber \\
&& \Rightarrow 
r_A \geq \delta \label{e:r_A}
\end{eqnarray}
where $r_A$ is the
radius of the sphere inscribed in simplex $A =\mathcal{H}(W)$.
Recall that $W$ includes one (arbitrary) point from each $F_i'$, $1\leq i\leq d+1$.

			Recall that there are $(d+1)f$ points in $S$,
and up to $f$ of them are received from faulty processes. Consider two cases:
\begin{itemize}
\item There exists $k$,  $1\leq k\leq d+1$, such that all the faulty inputs are contained in $F'_k$:
Then $\pi'_k$ is the convex hull of a subset of non-faulty inputs. By Lemma \ref{radius_lemma},
 we have $r_A<r_{\pi'_k}$, where $r_{\pi'_k}$ is the radius of inscribed sphere of $\pi'_k$
(in $d-1$ dimensions). By Lemma \ref{e_max_lemma}, we know that $r_{\pi'_k}< \frac{\max_{e\in E'}\|e\|_2}{d-1}$, where $E'$ is the set of edges between vertices of $\pi'_k$. Since $\pi'_k$ consists of only non-faulty inputs, we have $r_{\pi'_k}< \frac{\max_{e\in E'}\|e\|_2}{d-1}\leq \frac{\max_{e\in E_+}\|e\|_2}{d-1}$. Therefore, $\delta\leq r_A<r_{\pi'_k}< \frac{\max_{e\in E_+}\|e\|_2}{d-1}$.
\item There does not exist $k$ such that all the faulty inputs are contained in $F'_k$:
Since there are at most $f$ inputs, and $|F'_i|=f$ for each $i$, it follows
that each $F'_i$ contains at least one non-faulty input.
For $1\leq i\leq d+1$, let
$u_i\in F'_i$ be a non-faulty input. By Claim \ref{claim:simplex},
$C = \mathcal{H}(\{u_i~|~1\leq i\leq d+1\})$ is a simplex. 
Let $r_C$ denote the radius of the sphere inscribed in $C$. Then,
by (\ref{e:r_A}), we have $\delta\leq r_C$. Since 
by Lemma \ref{e_max_lemma}, we know that $r_C< \frac{\max_{e\in E''}\|e\|_2}{d}$, where $E''$ is the set of edges between the vertices of $C$. Since vectices
of $C$ are all non-faulty inputs, we have $r_C< \frac{\max_{e\in E''}\|e\|_2}{d}\leq \frac{\max_{e\in E_+}\|e\|_2}{d}$. Therefore, $\delta\leq r_C< \frac{\max_{e\in E_+}\|e\|_2}{d}
< \frac{\max_{e\in E_+}\|e\|_2}{d-1}$.
			\end{itemize}
			
Therefore, we obtain an upper bound for $\delta$ (i.e., $\delta^*(S)$) as
			$$
			\delta^*(S)<\frac{\max_{e\in E_{+}}\|e\|_2}{d-1}
			$$
where $E_{+}$ is the set of edges between the inputs of non-faulty processes.		
		\end{itemize}

	\end{proof}
	

Now we present a conjecture for the bound for the remaining cases. First we show that $\delta$ does not decrease when we remove some inputs.

\begin{lemma}\label{lemma:remove_nodes}
	Let $d\geq 3$, $f\geq 2$ and $3f+1< n \leq (d+1)f$. Consider the set of $n$ inputs $S=\{a_1,\cdots,a_{n}\}$ obtained in Step 1 of algorithm ALGO. Let the set of $n-1$ inputs $S'$ be obtained by removing any one inputs from $S$. Then
	$$
	\delta^*(S)\leq \delta^*(S')
	$$
\end{lemma}
\begin{proof}
	Let $P_i$, $i=1,\cdots,\binom{n}{f}$ be the subsets containing
$n-f$ inputs from $S$. Then,
	$$
	\delta^*(S)=\min_{p\in \mathbb{R}^d} \max_{i=1,\cdots,\binom{n}{f}} dist(p,\mathcal{H}(P_i))
	$$
	
	Let $P'_i$, $i=1,\cdots,\binom{n-1}{f}$ be the subsets
	containing $n-1-f$ inputs from $S'$. Then,
	$$
	\delta^*(S')=\min_{p\in \mathbb{R}^d} \max_{i=1,\cdots,\binom{n-1}{f}} dist(p,\mathcal{H}(P'_i))
	$$
	It is clear that every $\mathcal{H}(P'_i)$ is contained in some $\mathcal{H}(P_i)$, and every $\mathcal{H}(P_i)$ contains some $\mathcal{H}(P'_i)$. Therefore $$\max_{i=1,\cdots,\binom{n-1}{f}}dist(p,\mathcal{H}(P'_i))\geq \max_{i=1,\cdots,\binom{n}{f}}dist(p,\mathcal{H}(P_i))$$ for any $p\in \mathbb{R}^d$. Therefore $\delta^*(S')\geq\delta^*(S)$.
	
\end{proof}
	
	\begin{conjecture}\label{conjecture:1}
%
Let $d\geq 3$, $f\geq 2$ and $3f+1\leq n < (d+1)f$. Consider the set of $n$ inputs $S=\{a_1,\cdots,a_{n}\}$ obtained in Step 1 of algorithm ALGO. Then,
$$\delta^*(S)<\frac{\max_{e\in E_+}\|e\|_2}{\lfloor n/f \rfloor -2}$$
where $E_+$ is the set of edges between pairs of non-faulty inputs in $S$.
	\end{conjecture}

	\subsubsection{Summary of Upper Bounds}

	As we can see from Theorems \ref{e_min_basic}, \ref{e_max_general} and Conjecture \ref{conjecture:1}, we can summarize the upper bounds for any $n\geq 3f+1$ and $f\geq 1$.
	
	\begin{table}[htp]
		\centering
		\caption{Summary of upper bounds}
		\label{table:summary}
		\begin{tabular}{|l|l|l|}
			\hline
			& $f=1$ & $f\geq2$ \\ \hline
		
		$n=(d+1)f$ & \multirow{2}{*}{$\min(\frac{\min_{e\in E_+}\|e\|_2}{2}, \frac{\max_{e\in E_+}\|e\|_2}{n-2})$(Theorem \ref{e_min_basic})} & $\frac{\max_{e\in E_+}\|e\|_2}{d-1}$(Theorem \ref{e_max_general}) \\ \cline{1-1} \cline{3-3} 
		 
		$3f+1 \leq n<(d+1)f$	&  & $\frac{\max_{e\in E_+}\|e\|_2}{\lfloor n/f \rfloor-2}$(Conjecture \ref{conjecture:1}) \\ \hline
		\end{tabular}
	\end{table}

	If Conjecture \ref{conjecture:1} is valid, we can give an uniform upper bound for this problem.
	
	\begin{conjecture}\label{conjecture:2}
Let $d\geq 3$, $f\geq 1$ and $3f+1\leq n \leq (d+1)f$. Consider the set of $n$ inputs $S=\{a_1,\cdots,a_{n}\}$ obtained in Step 1 of algorithm ALGO. Then,
$$\delta^*(S)<\frac{\max_{e\in E_+}\|e\|_2}{\lfloor n/f \rfloor -2},$$
where $E_+$ is the set of edges between pairs of non-faulty inputs in $S$.
	\end{conjecture}

	\subsection{General Upper Bounds for $(\delta, p)$-Relaxed Exact BVC}
	In this section, for general values of $p$, we derive the upper bounds for $\delta$ for $(\delta, p)$-Relaxed Exact BVC problem based on previous result on $(\delta, 2)$-Relaxed Exact BVC problem in synchronous systems. Recall that we require that input-dependent $\delta$ must be bounded as follows:
	$$
	\delta \leq \kappa(n,f,d,p) \max_{e\in E_+} \|e\|_p,
	$$
	where $\kappa(n,f,d,p)$ is a finite constant that may depend on number of processes $n$, number of failures $f$, dimension of the inputs $d$ and $L_p$ norm, but not on the inputs.

	\begin{theorem}[Holder's inequality\cite{hardy1952inequalities}]\label{inequality}
		For vector $x\in \mathbb{R}^d$, let $\|x\|_p$ denotes the $L_p$-norm of $x$. 
		For $1\leq r\leq p$,
		$$
		\|x\|_p\leq \|x\|_r\leq d^{(\frac{1}{r}-\frac{1}{p})} \|x\|_p.
		$$
	\end{theorem}
	
	
Let us denote by $\delta_p^*(S)$ the smallest value of $\delta$
for which $(\delta,p)$-consensus is achievable for given set of inputs
$S$ obtained in Step 1 of algorithm ALGO.
Thus, $\delta^*(S)$ defined previously equals $\delta_2^*(S)$.

	
	\begin{theorem}\label{uniform_general}
Let $p\geq 2$, $d\geq 3$, $f\geq 1$ and $3f+1\leq n \leq (d+1)f$. Consider the set of $n$ inputs $S=\{a_1,\cdots,a_{n}\}$ obtained in Step 1 of the algorithm above. Suppose $(\delta,2)$-Relaxed Exact BVC problem can be solved with
$$\delta_2^*(S)< \kappa(n,f,d,2)\max_{e\in E_+}\|e\|_2$$
where $\kappa(n,f,d,2)$ is a constant, and $E_+$ is the set of edges between pairs of non-faulty inputs in $S$.
Then $(\delta,p)$-Relaxed Exact BVC problem can be solved with
$$\delta_p^*(S)< d^{(\frac{1}{2}-\frac{1}{p})}\kappa(n,f,d,2)\max_{e\in E_+}\|e\|_p$$
	\end{theorem}
	\begin{proof}
		
		By Theorem \ref{inequality}, we know 
		$$\|e\|_2\leq d^{(\frac{1}{2}-\frac{1}{p})} \|e\|_p$$
		and $$\delta_p^*(S) \leq \delta_2^*(S).$$
		Hence 
		$$\delta_p^*(S)\leq \delta_2^*(S)<\kappa(n,f,d,2)\max_{e\in E_+}\|e\|_2\leq d^{(\frac{1}{2}-\frac{1}{p})}\kappa(n,f,d,2)\max_{e\in E_+}\|e\|_p$$

	\end{proof}
	\begin{conjecture}\label{conjecture:3}
		Let $p\geq 2$, $d\geq 3$, $f\geq 1$ and $3f+1\leq n \leq (d+1)f$. Consider the set of $n$ inputs $S=\{a_1,\cdots,a_{n}\}$ obtained in Step 1 of the algorithm above. Then $(\delta,p)$-Relaxed Exact BVC problem can be solved with
		$$\delta_p^*(S)< \frac{d^{(\frac{1}{2}-\frac{1}{p})}}{\lfloor n/f \rfloor -2}\max_{e\in E_+}\|e\|_p$$
		where $E_+$ is the set of edges between pairs of non-faulty inputs in $S$.
	\end{conjecture}
	This conjecture is directly obtained from Conjecture \ref{conjecture:2} and Theorem \ref{uniform_general}.

\section{Input-Dependent $\delta$ for $(\delta,p)$-Relaxed Consensus in Asynchronous Systems}
In this section, we show that in asynchronous systems, $(\delta,p)$-Relaxed Approximate BVC problem is solvable with fewer than $(d+2)f+1$ processes. We propose an algorithm called \textsl{Relaxed Verified Averaging Algorithm} based on \textsl{Verified Averaging Algorithm} \cite{tseng2013byzantine}, and derive similar bounds for $\delta$ based on results in the synchronous case. Similar to synchronous case, we also require the input-dependent $\delta$ be bounded as 
$$\delta\leq \kappa'(n,f,d,p) \max_{e\in E_+}\|e\|_p$$
where $\kappa'(n,f,d,p)$ is a parameter which has the same definition with $\kappa(n,f,d,p)$ in the synchronous case. Similar to synchronous case, we define $\delta_p^*(S)$ to be the optimal $\delta$ such that the $(\delta,p)$-Relaxed Approximate BVC problem is solvable.

The algorithm we propose is based on \textsl{Verified Averaging Algorithm}. More specifically, we only need to modify Function $H(\mathcal{V},t)$ in \textsl{Verified Averaging Algorithm} as the following.

\begin{definition}{Function $H_{(\delta,p)}(\mathcal{V},t)$}
	\begin{enumerate}
		\item Define multiset $X:=\{h|(h,j,t-1)\in\mathcal{V}\}$.
		\item If $t=0$ then \textbf{hull}$:=\bigcap_{C\subseteq X, |C|=|X|-f}\mathcal{H}_{(\delta,p)}(C)$. Deterministically pick a point \textbf{temp} from \textbf{hull}.
		\item If $t>0$ then \textbf{temp}$:=\sum_{h_i\in X}\frac{1}{|X|}h_i$.
		\item Return \textbf{temp}.
	\end{enumerate}
\end{definition}

Notice that by step $2$ in the definition of Function $H_{(\delta,p)}(\mathcal{V},t)$, the return value \textbf{temp} of round $t=0$ is a single vector, instead of a convex hull in \textsl{Verified Averaging Algorithm}. Then by step $3$ in the definition of Function $H_{(\delta,p)}(\mathcal{V},t)$, the return value \textbf{temp} will always be a single vector for any round $t\geq 1$.


\subsection{Algorithm}
\textsl{Relaxed Verified Averaging Algorithm} is basically \textsl{Verified Averaging Algorithm} with Function $H_{(\delta,p)}(\mathcal{V},t)$ instead of $H(\mathcal{V},t)$. The details of the algorithm can be found in \cite{tseng2013byzantine}.

Many results for \textsl{Verified Averaging Algorithm} are also valid for \textsl{Relaxed Verified Averaging Algorithm}. Hence in this section, we only sketch the proof for the correctness of the algorithm. Notice that necessary condition $n\geq 3f+1$ is also required similar to the synchronous case, since for \textsl{Relaxed Verified Averaging Algorithm}, $n\geq 3f+1$ is necessary to guarantee the correctness of \textsl{reliable broadcast} \cite{bracha1987asynchronous} used in the algorithm.

\subsection{Upper Bounds on $\delta$}
We have the following result for the asynchronous case.


\begin{theorem}\label{thm:asyn}
	Suppose $(\delta,p)$-Relaxed Exact BVC problem can be solved with
	$$\delta_p^*(S)< \kappa(n,f,d,p)\max_{e\in E_+}\|e\|_p$$
	where $\kappa(n,f,d,p)$ is a constant defined previously, and $E_+$ is the set of edges between pairs of non-faulty inputs in $S$.
	
	Then $(\delta,p)$-Relaxed Approximate BVC can be solved with
	$$
	\delta_p^*(S)<\kappa(n-f,f,d,p)\max_{e\in E_+}\|e\|_p
	$$
	where $\kappa(n,f,d,p)$, $S$ and $E_+$ are defined above.

\end{theorem}


\begin{proof}

	First we prove that \textsl{Relaxed Verified Averaging Algorithm} solves $(\delta,p)$-Relaxed Approximate BVC, namely satiesfies $(\delta,p)$-Relaxed Validity and $\epsilon$-Agreement properties after a large enough number of asynchronous rounds, as long as $H_{(\delta,p)}(\mathcal{V},0)$ is non-empty.

{\em For $(\delta,p)$-Relaxed Validity:}

In round $0$, by Reliable Broadcast and the definition of $H_{(\delta,p)}(\mathcal{V},0)$, we know that \textbf{hull} in Fuction $H_{(\delta,p)}(\mathcal{V},0)$ is a non-empty subset of $(\delta,p)$-relaxed convex hull of non-faulty inputs. Hence by similar proof in Lemma 5 \cite{tseng2013byzantine}, $h_i[0]$ is $(\delta,p)$-relaxed valid (i.e., in the $(\delta,p)$-relaxed convex hull of non-faulty inputs) for any process $i$ verified in round $0$.

In round $t>1$, by similar argument in Theorem 2 \cite{tseng2013byzantine}, for process $i$ is verified in round $t$, we have $h_i[t]$ being convex combination of $h_j[0]$'s where process $j$ is verified in round $0$. Thus $h_i[t]$ is $(\delta,p)$-relaxed valid for all non-faulty process $i$ in round $t$. By induction, the algorithm satisfies $(\delta,p)$-Relaxed Validity condition.

{\em For $\epsilon$-Agreement:}

Notice that for round $t= 0$, the return value of $H_{(\delta,p)}(\mathcal{V},0)$ is a point.
Since a point is a special case of a convex hull and the algorithm remains identical to \textsl{Verified Averaging Algorithm} except round $0$, the argument in Theorem 2 \cite{tseng2013byzantine} for $\epsilon$-agreement also applies.

	Now We can show that when $H_{(\delta,p)}(\mathcal{V},0)$ is non-empty, we have the upper bound $\delta_p^*(S)<\kappa(n-f,f,d,p)\max_{e\in E_+}\|e\|_p$.
	
	Consider \textbf{hull}$=\bigcap_{C\subseteq X,|C|=|X|-f}\mathcal{H}_{(\delta,p)}(C)$ in the definition of $H_{(\delta,p)}(\mathcal{V},0)$. Since $X$ contains at most $f$ faulty inputs, and $\max_{e\in E_{X+}}\|e\|_p\leq \max_{e\in E_{+}}\|e\|_p$ where $E_{X+}$ is the set of inputs of non-faulty processes in $X$, we know that when \textbf{hull} is non-empty, we have the bound $\delta^*_p(S)<\kappa(|X|,f,d,p)\max_{e\in E_+}\|e\|_p$ by assumption. Notice that in the algorithm when $H_{(\delta,p)}(\mathcal{V},0)$ is called, we have $|X|=|\mathcal{V}|\geq n-f$. Hence when $H_{(\delta,p)}(\mathcal{V},0)$ is non-empty, we have $\delta_p^*(S)<\kappa(n-f,f,d,p)\max_{e\in E_+}\|e\|_p$

Hence $(\delta,p)$-Relaxed Approximate BVC can be solved with 
$$
\delta_p^*(S)<\kappa(n-f,f,d,p)\max_{e\in E_+}\|e\|_p
$$

\end{proof}

\begin{conjecture}\label{conjecture:4}
Let $p\geq 2$, $d\geq 3$, $f\geq 1$ and $3f+1\leq n \leq (d+2)f$. Then $(\delta,p)$-Relaxed Approximate BVC problem can be solved with
$$\delta^*_p(S)<\frac{d^{(\frac{1}{2}-\frac{1}{p})}}{\lfloor n/f \rfloor -3}\max_{e\in E_{+}}\|e\|_p$$
where $E_+$ is the set of edges between pairs of non-faulty inputs in $S$.
\end{conjecture}

	This conjecture is directly obtained from Conjecture \ref{conjecture:3} and Theorem \ref{thm:asyn}.


\section{Summary}
This paper studies \textsl{k-Relaxed Byzantine vector consensus} and \textsl{$(\delta,p)$-Relaxed Byzantine vector consensus} with constant $\delta$, and $\delta$ dependent on the inputs, respectively. For the first two relaxed version of Byzantine vector consensus problem, the tight necessary and sufficient conditions remain unchanged compared to the original problem in both synchronous and asynchronous systems. For the third relaxed version, the tight conditions can be relaxed. We establish partial results concerning the upper bounds of $\delta$ in terms of different number of processes, and propose a conjecture for one remaining case.

	
	\bibliographystyle{abbrv}
	\bibliography{ref}  

\begin{thebibliography}{10}

\bibitem{abraham2005optimal}
I.~Abraham, Y.~Amit, and D.~Dolev.
\newblock Optimal resilience asynchronous approximate agreement.
\newblock In {\em Principles of Distributed Systems}, pages 229--239. Springer,
  2005.

\bibitem{akira2014radii}
A.~Akira~Toda.
\newblock Radii of the inscribed and escribed spheres of a simplex.
\newblock {\em International Journal of Geometry}, 3(2), 2014.

\bibitem{barany1982generalization}
I.~B{\'a}r{\'a}ny.
\newblock A generalization of carath{\'e}odory's theorem.
\newblock {\em Discrete Mathematics}, 40(2):141--152, 1982.

\bibitem{bracha1987asynchronous}
G.~Bracha.
\newblock Asynchronous byzantine agreement protocols.
\newblock {\em Information and Computation}, 75(2):130--143, 1987.

\bibitem{danzer1963helly}
L.~Danzer, B.~Gr{\"u}nbaum, and V.~Klee.
\newblock Helly's theorem and its relatives, 1963.

\bibitem{dolev1986reaching}
D.~Dolev, N.~A. Lynch, S.~S. Pinter, E.~W. Stark, and W.~E. Weihl.
\newblock Reaching approximate agreement in the presence of faults.
\newblock {\em Journal of the ACM (JACM)}, 33(3):499--516, 1986.

\bibitem{fischer1990easy}
M.~J. Fischer, N.~A. Lynch, and M.~Merritt.
\newblock {\em Easy impossibility proofs for distributed consensus problems}.
\newblock Springer, 1990.

\bibitem{fischer1985impossibility}
M.~J. Fischer, N.~A. Lynch, and M.~S. Paterson.
\newblock Impossibility of distributed consensus with one faulty process.
\newblock {\em Journal of the ACM (JACM)}, 32(2):374--382, 1985.

\bibitem{hardy1952inequalities}
G.~H. Hardy, J.~E. Littlewood, and G.~P{\'o}lya.
\newblock {\em Inequalities}.
\newblock Cambridge university press, 1952.

\bibitem{herlihy2014computing}
M.~Herlihy, S.~Rajsbaum, M.~Raynal, and J.~Stainer.
\newblock Computing in the presence of concurrent solo executions.
\newblock In {\em LATIN 2014: Theoretical Informatics}, pages 214--225.
  Springer, 2014.

\bibitem{kothe1983topological}
G.~K{\"o}the and G.~K{\"o}the.
\newblock {\em Topological vector spaces}.
\newblock Springer, 1983.

\bibitem{lamport1982byzantine}
L.~Lamport, R.~Shostak, and M.~Pease.
\newblock The byzantine generals problem.
\newblock {\em ACM Transactions on Programming Languages and Systems (TOPLAS)},
  4(3):382--401, 1982.

\bibitem{lynch1989hundred}
N.~Lynch.
\newblock A hundred impossibility proofs for distributed computing.
\newblock In {\em Proceedings of the eighth annual ACM Symposium on Principles
  of distributed computing}, pages 1--28. ACM, 1989.

\bibitem{mendes2013multidimensional}
H.~Mendes and M.~Herlihy.
\newblock Multidimensional approximate agreement in byzantine asynchronous
  systems.
\newblock In {\em Proceedings of the forty-fifth annual ACM symposium on Theory
  of computing}, pages 391--400. ACM, 2013.

\bibitem{tseng2013byzantine}
L.~Tseng and N.~Vaidya.
\newblock Byzantine convex consensus: An optimal algorithm.
\newblock {\em arXiv preprint arXiv:1307.1332}, 2013.

\bibitem{tseng2014asynchronous}
L.~Tseng and N.~H. Vaidya.
\newblock Asynchronous convex hull consensus in the presence of crash faults.
\newblock In {\em Proceedings of the 2014 ACM symposium on Principles of
  distributed computing}, pages 396--405. ACM, 2014.

\bibitem{tverberg1966generalization}
H.~Tverberg.
\newblock A generalization of radon's theorem.
\newblock {\em J. London Math. Soc}, 41(1):123--128, 1966.

\bibitem{vaidya2014iterative}
N.~H. Vaidya.
\newblock Iterative byzantine vector consensus in incomplete graphs.
\newblock In {\em Distributed Computing and Networking}, pages 14--28.
  Springer, 2014.

\bibitem{vaidya2013byzantine}
N.~H. Vaidya and V.~K. Garg.
\newblock Byzantine vector consensus in complete graphs.
\newblock In {\em Proceedings of the 2013 ACM symposium on Principles of
  distributed computing}, pages 65--73. ACM, 2013.

\end{thebibliography}

\newpage


\appendix

\section{Proof of Lemma \ref{lemma:imposs}}\label{app:a}
	
	\begin{proof}
The proof is essentially identical to the impossibility proof of scalar Byzantine consensus in \cite{lynch1989hundred}. We first show $(\delta, p)$-consensus is impossible for $n\leq 3$ processes with one faulty process.

Let $\mathbf{0}^d$ denotes the vector in dimension $d$ and all its elements are $0$.
Let $\mathbf{1}^d$ denotes the vector in dimension $d$ and all its elements are $1$.

Suppose that a correct algorithm $A$ exists.
		
Suppose that processes $p,q,r$ can solve $(\delta, p)$-relaxed Byzantine vector consensus, where $\delta
\leq \kappa`\|x-y\|_p$ and $n=3$, $f=1$. Consider three scenarios, namely an execution of the system, in Figure \ref{fig:scenario}.
		
		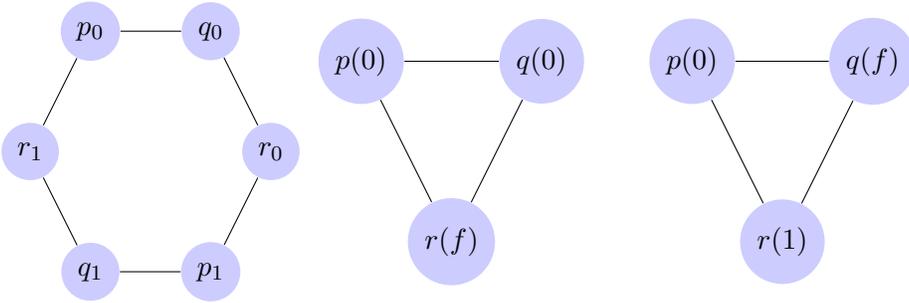
\begin{figure}[htp]
			\centering
			\begin{tikzpicture}
			[scale=.8,auto=left,every node/.style={circle,fill=blue!20}]
			\node (p0) at (2,7) {$p_0$};
			\node (q0) at (4,7) {$q_0$};
			\node (r0) at (5,5) {$r_0$};
			\node (q1) at (2,3) {$q_1$};
			\node (p1) at (4,3) {$p_1$};
			\node (r1) at (1,5) {$r_1$};
			
			\foreach \from/\to in {p0/q0,r0/q0,p1/r0,q1/p1,r1/q1,p0/r1}
			\draw (\from) -- (\to);
			
			\node (p) at (6.5,6.5) {$p(0)$};
			\node (q) at (9.5,6.5) {$q(0)$};
			\node (r) at (8,3.5) {$r(f)$};
			
			\foreach \from/\to in {p/q,r/q,p/r}
			\draw (\from) -- (\to);
			
			\node (p) at (12,6.5) {$p(0)$};
			\node (q) at (15,6.5) {$q(f)$};
			\node (r) at (13.5,3.5) {$r(1)$};
			
			\foreach \from/\to in {p/q,r/q,p/r}
			\draw (\from) -- (\to);
			
			\end{tikzpicture}
			\caption{Scenarios for impossibility proof}
			\label{fig:scenario}
		\end{figure}

		In the first scenario $A$, the system consists two copies of $p,q,r$, namely $p_0,q_0,r_0,p_1,q_1,r_1$, joined into a ring. The processes $p_0,q_0,r_0$ start with initial value $\mathbf{0}^d$, and the processes $p_1,q_1,r_1$ start with initial value $\mathbf{1}^d$.
		
		Consider the second scenario $B$, consisting of one copy of $p,q,r$. Both $p,q$ start with initial value $\mathbf{0}^d$, and $r$ is faulty. Since $r$ can behave arbitraryly, it can send $p$ exactly what $r_1$ send to $p_0$ in scenario $A$, and also send $q$ what $r_0$ send to $q_0$ in scenario $A$. According to the $(\delta,p)$-Relaxed Validity condition, the decision vector at each non-faulty process must be in the \textsl{$(\delta,p)$-relaxed convex hull} of the input vectors at the non-faulty processes. Notice that here $\delta\leq\kappa\|\mathbf{0}^d-\mathbf{0}^d\|_p=0$, the convex hull of feasible outputs is not relaxed. Therefore the decision vector of $p$ and $q$ must be $\mathbf{0}^d$ in scenario $B$. Therefore, $p_0$ and $q_0$ will also decide on $\mathbf{0}^d$ in scenario $A$. By similar argument, $q_1$ and $r_1$ will decide on $\mathbf{1}^d$ in scenario $A$.
		
		Now consider the third scenario $C$, where there are one copy of $p,q,r$. $p$ starts with $\mathbf{0}^d$, $r$ starts with $\mathbf{1}^d$, and $q$ is faulty. Let $q$ send $p$ exactly what $q_0$ send to $p_0$ in scenario $A$, and also send $r$ what $q_1$ send to $r_1$ in scenario $A$. $p$ and $r$ has to decide on a single output in scenario $C$, and so does $p_0$ and $r_1$ in scenario $A$. However, it contradicts with the previous argument that $p_0$ must decide on $\mathbf{0}^d$ and $r_1$ must decide on $\mathbf{1}^d$. Therefore $(\delta, p)$-relaxed Byzantine vector consensus is impossible for $n=3$, $f=1$.
		
		As for $f>1$, we can use simulation approach to show $(\delta, p)$-relaxed Byzantine vector consensus is impossible for $n=(d+1)f$, completing the proof.
	\end{proof}
	
	Since it is impossible to solve $(\delta, p)$-relaxed Byzantine vector consensus for $n\leq 3f$, we only consider the case where $n\geq 3f+1$.

\section{Proof of Theorem \ref{thm:k:async}}\label{app:b}
\begin{proof}
	Since $2\leq k\leq d-1$, $d\geq 3$.
	~\newline
	{\em Sufficiency:}
	By Theorem \ref{thm:approximate}, and due to the equivalence of the original Approximate BVC and
	$d$-Relaxed Approximate consensus, for $d\geq 2$, $n\geq (d+2)f+1$ is sufficient for $d$-Relaxed Approximate BVC.
	Then by Lemma \ref{lemma:k:4}, this condition is also sufficient for $k$-Relaxed Approximate BVC where $2\leq k\leq d-1$.
	
	{\em Necessity:}
	Similar to the synchronous case, we first prove that $n\geq d+3$ is necessary for $f=1$ and $k=2$ by contradition. Suppose that $n=d+2$ and $2$-Relaxed Approximate BVC is achievable using a certain algorithm when $f=1$. Similar to the synchronous case, we assume that all processes follow the specified algorithm. 
	
	Let the $i^{th}$ column of the following $d\times (d+2)$ matrix $S$ be an input vector of the $i^{th}$ process, where $0<2\epsilon<\gamma$. We can show that these $d+2$ inputs lead to an empty solution set.
	
	\begin{equation*}
	S=\begin{pmatrix}
	\gamma & 0 & \cdots & \cdots & 0 &-\gamma & 0\\
	2\epsilon& \gamma & 0 & \cdots & 0 &-\gamma & 0\\
	\vdots&  \ddots & \ddots & \ddots & \vdots &\vdots & \vdots\\
	2\epsilon & \cdots & 2\epsilon & \gamma & 0 & -\gamma & 0\\
	2\epsilon & \cdots & \cdots & 2\epsilon & \gamma & -\gamma & 0\\
	\end{pmatrix}
	\end{equation*}
	In column $i$, $1\leq i\leq d$, the first $i-1$ elements equal 0, the $i$-th element equals $\gamma$, and the rest of the elements equal $2\epsilon$. In column $d+1$, all elements are $-\gamma$. In column $d+2$, all elements are $0$.
	
	Let $s_i$ denote the $i$-th column of matrix $S$ above, that is, the
	input of $i$-th process. Here we use an argument similar to that in the proof of Theorem 4 in \cite{vaidya2013byzantine}. Define $S^{j}=\{s_i:1\leq s_i\leq d+1$ and $i\neq j \}$, $S^{d+2}=\{s_i:1\leq s_i\leq d+1\}$. 
	
	Since a correct algorithm must tolerate one failure, process $i$ must terminate in finite steps even when process ${d+2}$ takes no step (but other processes do take steps). When process $i$ terminates, it cannot distinguish the following $d+1$ cases:
	\begin{itemize}
		\item
		Process $d+2$ has crashed: In order to satisfy the \textsl{k-Relaxed Validity} condition, the output of process $p_i$ must be in the \textsl{k-relaxed convex hull} of input $s_1,\cdots,s_{d+1}$, namely $H_k(S^{d+2})$.
		\item
		Process $j$ is faulty, process $d+2$ is slow: In order to satisfy the \textsl{k-Relaxed Validity} condition, process $i$ cannot trust the process $j$. Therefore the output of process $i$ must be in $H_k(S^{j})$. Since process $j$ may be any process other
		than process $i$ and $d+2$, its output must be in 
		$\bigcap\limits_{j\neq i, 1\leq j \leq d+1}H_k(S^{j})$.
	\end{itemize}
	
	Since $\bigcap\limits_{j\neq i, 1\leq j \leq d+1}H_k(S^{j})\subseteq H_k(S^{d+2})$, the output of process $i$ must be in
	$$
	\bigcap_{j\neq i, 1\leq j \leq d+1}H_k(S^{j})=\bigcap_{\substack{D\in D_k\\j\neq i, 1\leq j \leq d+1}}g_D^{-1}(\mathcal{H}(g_D(S^j)))
	$$
	We denote the above set that must contain
	the output of process $i$ as
	$$
	\Psi_i(S)=\bigcap_{\substack{D\in D_k\\j\neq i, 1\leq j \leq d+1}}g_D^{-1}(\mathcal{H}(g_D(S^j)))
	$$
	We will consider the output set for several different processes now.
	\begin{enumerate}
		\item Consider the output set $\Psi_1(S)$ of process 1:
		\begin{itemize}
			\item Observation 1: First consider $D=\{d-1,d\}$ and $j=d$. For each vectors in $S^{d}$, the $d$-th coordinate is less than or equal to $2\epsilon$ since $\gamma>0$ and $\epsilon>0$. Hence the $d$-th coordinate of vectors in $\Psi_1(S)$ must be less than or equal to $2\epsilon$. Then consider $D=\{d-1,d\}$ and $j=d+1$. Similarly, the $d$-th coordinate of vectors in $\Psi_1(S)$ must be greater than or equal to $2\epsilon$. Therefore the $d$-th coordinate of vectors in $\Psi_1(S)$ must be $2\epsilon$.
			
			\item Observation 2: Consider $D=\{t,t+1\}$ where $1\leq\ t\leq d-1$ and $j=t+1$. 
			As we can see, the $t$-th coordinate of vectors in $S^{t+1}$ is greater than or equal to the $t+1$-th coordinate for $1\leq\ t\leq d-1$, since $2\epsilon<\gamma$. Hence the $t$-th coordinate of vectors in $\Psi_1(S)$ must be greater than or equal to the $t+1$-th coordinate for $1\leq\ t\leq d-1$.
			
			\item Observation 3: Combining Obsevation 1 and 2, we have the first coordinate of vectors in $\Psi_1(S)$ must be greater than or equal to $2\epsilon$.
		\end{itemize}
		
		\item Consider the output set $\Psi_2(S)$ of process 2:
		\begin{itemize}
			\item Observation 4: First consider $D=\{1,2\}$ and $j=1$. For each vectors in $S^{1}$, the first coordinate is less than or equal to $0$, since $\gamma>0$. Hence the first coordinate of vectors in $\Psi_2(S)$ must be less than or equal to $0$. Then consider $D=\{1,2\}$ and $j=d+1$. Since the first coordinate of
			all vectors in $S^{d+1}$ is non-negative, the
			first coordiate of all vectors in $\Psi_2(S)$ must also be greater than or equal to $0$. 
			Combining the two observations,
			the first coordinate of vectors in $\Psi_2(S)$ must be $0$.
		\end{itemize}
	\end{enumerate}
	
	By Observation 3 and Observation 4, $\epsilon$-Agreement is violated since we have $\|v_1-v_2\|_\infty\geq 2\epsilon$, for any $v_1\in \Psi_1(S)$ and $v_2\in \Psi_2(S)$. Therefore we have proved $n=d+2$ is not sufficient for $f=1$, $k=2$. 
	
	As for $f>1$, we can use simulation approach to show $n=(d+2)f$ is not sufficient \cite{lamport1982byzantine}. Therefore, $n \geq (d + 2)f + 1$ is necessary for $f \geq 1$, $k=2$.

	Now, for any vector $x$, $\|x\|_r\leq \|x\|_p$ when $1\leq p\leq r$ \cite{kothe1983topological}. This
	implies that if $\epsilon$-agreement is not achieved under the $L_\infty$-norm,
	then $\epsilon$-agreement is also not achieved under the $L_p$-norm, where $1\leq p<\infty$.
	So the above bound on $n$ holds for any $L_p$-norm, $p\geq1$.
	
\end{proof}

\section{Proof of Theorem \ref{thm:delta:async}}\label{app:c}
\begin{proof}
	For $d=1$, the bound $n\geq 3f+1$ is tight for $(\delta,p)$-Relaxed Approximate BVC \cite{fischer1990easy,abraham2005optimal}. Therefore we only consider the case $d\geq2$.
	
	{\em Sufficiency:}
	By Theorem \ref{thm:exact}, and due to the equivalence of the original Approximate BVC and
	$(0, p)$-Relaxed Approximate BVC, for $d\geq 2$ and $1\leq p$, $n\geq (d+2)f+1$ is sufficient for $(0, p)$-Relaxed Approximate BVC. Then by Lemma \ref{lemma:delta:4}, this condition is also sufficient for $(\delta, p)$-Relaxed Approximate BVC where $0<\delta<\infty$.
	
	{\em Necessity:}	
	Similar to the synchronous case, we first prove the necessary condition for \textsl{$(\delta, \infty)$-Relaxed Approximate BVC}.
	
	We first prove that $n\geq d+3$ is necessary for $f=1$ case. The proof of necessity is by contradiction. Suppose that $n=d+2$ and $(\delta,\infty)$-Relaxed Approximate BVC is achievable using a certain algorithm. 
	
	Analogous to the proof of Theorem \ref{thm:k:sync}, we assume that any faulty process follows the algorithm correctly.
	Let the $i^{th}$ column of the following $d\times (d+2)$ matrix $S$ be an input vector of the $i^{th}$ process. We show that these $d+2$ inputs lead to empty output when $x>2d\delta+\epsilon$.
	\begin{equation*}
	S=\begin{pmatrix}
	x & 0 & \cdots & \cdots & 0 & 0& 0\\
	0& x & 0 & \cdots & 0 & 0& 0\\
	\vdots&  \ddots & \ddots & \ddots & \vdots &\vdots&\vdots\\
	0 & \cdots & 0 & x & 0 & 0& 0\\
	0 & \cdots & \cdots & 0 & x & 0& 0\\
	\end{pmatrix}
	\end{equation*}
	
	By arguments similar to the proof of Theorem \ref{thm:k:async}, the output of process $i$ must be in
	$$
	\Psi_i(S)=\bigcap_{j\neq i, 1\leq j \leq d+1}H_{(\delta,\infty)}(S^{j})
	$$
	
	Let us consider the output for different processes:
	\begin{enumerate}
		\item Consider the output set $\Psi_1(S)$ of process 1:
		\begin{itemize}
			\item Observation 1: Consider $2\leq j\leq d$. Since the $j$-th element of all vectors in $S^j$ is 0, the $j$-th element of the vectors in $\Psi_1(S)$ is less or equal than $\delta$, due to the definition of $(\delta,\infty)$-Relaxed Validity.
			
			\item Observation 2: Consider $j=d+1$.
			Recall that the vectors
			in $H_{(\delta,\infty)}(T)$ are within distance $\delta$ (where the
			distance is measured using the $L_\infty$ norm) of the convex hull $\mathcal{H}(T)$.
			For a given vector $v$ in $\mathcal{H}(T)$,
			let $\alpha_t$ be the weight attached to input $s_t$
			to obtain $v$ (i.e., $v$ is a weighted linear combination of $s_t$'s
			with weights being $\alpha_t$'s). By Observation 1, in order to have non-empty $\Psi_1(S)$, for $2\leq t\leq d$, we must have
			$$
			\alpha_t x-\delta\leq \delta
			$$
			that is, $\alpha_t\leq 2\delta/x$. Hence the weight of $s_1$ in the original convex hull must be larger than or equal to $1-(d-1)2\delta/x$. Therefore the first element of the vectors in $\Psi_1(S)$ is $\geq x\times (1-(d-1)2\delta/x)-\delta=x-(2d-1)\delta$.
		\end{itemize}
		
		\item Consider the output set $\Psi_2(S)$ of process 2:
		\begin{itemize}
			\item Observation 3: Consider $j=1$. Since the first element of all vectors in $S^1$ is $0$, the first element of the vectors in $\Psi_2(S)$ is less or equal than $\delta$, due to the definition of $(\delta,\infty)$-Relaxed Validity.
		\end{itemize}
	\end{enumerate}
	
	According to the assumption, we have $x-(2d-1)\delta-\delta>\epsilon$. Hence we have $\|v_1-v_2\|_\infty>\epsilon$, for any $v_1\in \Psi_1(S)$ and $v_2\in \Psi_2(S)$. Therefore $n=d+2$ is not sufficient for $f=1$. 
	
	For $f>1$, we can use the simulation approach to show $n=(d+2)f$ is not sufficient \cite{lamport1982byzantine}. Therefore, $n \geq (d + 2)f + 1$ is necessary for $f \geq 1$, completing the proof for \textsl{$(\delta, \infty)$-Relaxed Approximate BVC}.
	By an argument similar to the synchronous case, the above bound also
	extends to $(\delta,p)$-relaxed approximate BVC, $p\geq 1$.
	
	Since if $\epsilon$-agreement is not achieved under the $L_\infty$-norm,
	then $\epsilon$-agreement is also not achieved under the $L_p$-norm, where $1\leq p<\infty$.
	So the above bound on $n$ holds for any $L_p$-norm, $p\geq1$.
\end{proof}

\end{document}